\documentclass{article}

\PassOptionsToPackage{numbers, compress}{natbib}

\newcommand{\indist}{\overset{c}{\equiv}}
\usepackage[preprint]{neurips_2020}
\usepackage{amsthm}
\usepackage[utf8]{inputenc} 
\usepackage[T1]{fontenc}    
\usepackage{hyperref}       
\usepackage{url}            
\usepackage{booktabs}       
\usepackage{amsfonts}       
\usepackage{nicefrac}       
\usepackage{microtype}      
\usepackage{enumitem}
\usepackage{multirow}
\usepackage{makecell}
\usepackage{graphicx}
\usepackage{caption}
\usepackage{courier}
\usepackage{multirow}
\usepackage{color}
\usepackage{subcaption}
\usepackage{todonotes}
\usepackage{amssymb}
\usepackage{textcomp,gensymb}
\usepackage{algorithm}
\usepackage{amsmath,nccmath}
\usepackage{algorithmic}
\usepackage{paralist}
\usepackage{cleveref}
\usepackage{amsthm}
\usepackage{mdframed}

\theoremstyle{definition}
\newtheorem{definition}{Definition}[section]
\newtheorem{theorem}{Theorem}

\title{Federated Dynamic GNN with Secure Aggregation}

\author{%
  Meng Jiang \\
  University of Notre Dame \\
  \texttt{mjiang2@nd.edu} \\
   \And
  Taeho Jung \\
  University of Notre Dame \\
  \texttt{tjung@nd.edu} \\
   \AND
  Ryan Karl \\
  University of Notre Dame \\
  \texttt{rkarl@nd.edu} \\
   \And
  Tong Zhao \\
  University of Notre Dame \\
  \texttt{tzhao2@nd.edu} \\
}

\begin{document}

\maketitle

\begin{abstract}

Given video data from multiple personal devices or street cameras, can we exploit the structural and dynamic information to learn dynamic representation of objects for applications such as distributed surveillance, \emph{without} storing data at a central server that leads to a violation of user privacy? In this work, we introduce Federated Dynamic Graph Neural Network (Feddy), a distributed and secured framework to learn the object representations from multi-user graph sequences:
\emph{i)} It aggregates structural information from nearby objects in the current graph as well as dynamic information from those in the previous graph. It uses a self-supervised loss of predicting the trajectories of objects.
\emph{ii)} It is trained in a federated learning manner. The centrally located server sends the model to user devices. Local models on the respective user devices learn and periodically send their learning to the central server without ever exposing the user's data to server.
\emph{iii)} Studies showed that the aggregated parameters could be inspected though decrypted when broadcast to clients for model synchronizing, after the server performed a weighted average. We design an appropriate aggregation mechanism of secure aggregation primitives that can protect the security and privacy in federated learning with scalability.
Experiments on four video camera datasets (in four different scenes) as well as simulation demonstrate that Feddy achieves great effectiveness and security.

\end{abstract}


\section{Introduction}

Distributed surveillance systems have the ability to detect, track, and snapshot objects moving around in a certain space \cite{valera2005intelligent}. The system is composed of several smart cameras (e.g., personal devices or street cameras) that are equipped with a high-performance onboard computing and communication infrastructure and a central server that processes data \cite{bramberger2006distributed}. Smart survelliance is getting increasingly popular as machine learning technologies (ML) become easier to use \cite{chen2019distributed}. Traditionally, Convolutional Neural Networks (CNNs) were integrated in the cameras and employed to identify and segment objects from video streams \cite{yazdi2018new}. Then raw features (\emph{e.g.}, color, position) of objects were sent to the server and used to train a learning model for tracking and/or anomaly detection. In order to render a model for making accurate decisions, it has to learn \emph{higher-level features and patterns} from multi-user or multi-source video data. The model is expected to obtain complex patterns such as vehicles slowing down at traffic circles, stopping at traffic lights, and bicyclists cutting in and out of traffic in an \emph{unsupervised} way.

Graph neural networks (GNNs) have been applied to capture deep patterns in vision data across different problems such as object detection \cite{yazdi2018new}, situation recognition \cite{li2017situation}, and traffic forecasting \cite{li2017diffusion,yu2017spatio}. When objects are identified, a video frame can be presented as a graph where nodes are the objects and links describe the spatial relationship between objects. Yet there are three challenges we identify when designing and deploying GNN models in the distributed surveillance system.

First, an effective, annotation-free task is desired for training a GNN model on long graph sequences. The models are expected to preserve the deep spatial and dynamic moving patterns as described above in the latent representations of objects \cite{manessi2017dynamic,zhou2018dynamic,sankar2018dynamic,pareja2020evolvegcn}. Second, collecting raw video or graph data from a large number of devices and training on a central server would not be a feasible solution, though we need one global model \cite{konevcny2016federated,bonawitz2019towards,yang2019federated}. After being deployed on real hardware, the GNN model should be updated locally at each device as needed, or ``fine-tuned'' with newly-collected local data, to adapt to unique scenarios so it can make decisions on-the-fly. When there are sufficient resources (\emph{e.g.}, communication bandwidth), the model updates can be shared among the devices to let them agree on the global model, which captures various deep patterns learned from individual devices or users. Third, distributed systems are vulnerable to various inference attacks \cite{fredrikson2015model,truex2019demystifying,melis2019exploiting}. Namely, adversaries who observe the updates of individual models are able to infer significant information about the individual training datasets (\emph{e.g.}, distribution or even samples of the training datasets) \cite{abadi2016deep,nasr2019comprehensive}. When the central server is compromised, individual datasets are compromised as well with the inference attacks.

In this work, we propose a novel approach called \underline{Fed}erated \underline{Dy}namic Graph Neural Network (Feddy) to address the three challenges. Generally, it is an unsupervised, distributed, secured framework to learn the object representations in graph sequences from several users or devices.

\emph{i)} We define an MSE loss on the task of future position prediction to train the GNN model. Given node attributes and relational links in a graph sequence before time $t$, the model generates the latent representation of nodes in the graph of time $t$ via neural aggregation functions, and it learns the parameters to predict the node positions at time $t+\Delta t$. In our study, $\Delta t$ is 5 seconds, i.e., 150 graphs as default for 30 fps video. The node attributes include objects' horizontal and vertical positions, object box size, and RGB colors in the center, on the left, right, top, and bottom of the box. The relational links include spatial relationships between nodes within a graph and the dynamic relationship of the same node in neighboring graphs. Spatial and dynamic patterns are preserved by this dynamic GNN model. The boxes were identified by CNN tools, so no human annotation is needed in the training process.

\emph{ii)} We use federated learning (FL) to train one dynamic GNN model across devices without exchanging training data. With FL, individual cameras compute the model updates (\emph{e.g.}, gradients, updated weights) locally, and only these model updates are shared with the central server (called \emph{parameter server}) who trains a global model using the aggregated updates collected from individual devices. The parameter server then shares the trained model with all other devices. FL provides a viable platform for state-of-the-art ML, and it is privacy-friendly because the training data never leaves individuals.

\emph{iii)} We employ secure aggregation to prevent inference attacks launched by malicious parameter servers. Secure aggregation is a primitive that allows a third party aggregator to efficiently compute an aggregate function (\emph{e.g.}, product, sum, average) over individuals' private input values. When this primitive is applied to FL, the parameter server can access the aggregated model only, and adversaries can no longer launch the aforementioned inference attacks to infer individuals' training data.

In our experiments, we use Stanford Drone Dataset\footnote{https://cvgl.stanford.edu/projects/uav\_data/} published by the Stanford Computational Vision Geometry Lab. The large-scale dataset collects images and videos of various types of agents (\emph{e.g.}, pedestrians, bicyclists, skateboarders, cars, buses, and golf carts) that navigate in university campus \cite{robicquet2016learning}. We use videos in four scenes such as ``bookstore,'' ``coupa,'' ``hyang,'' and ``little.'' Experimental results demonstrate that the proposed framework can consistently deliver good performance in an unsupervised, distributed, secured manner.

\section{Dynamic GNN for Learning Surveillance Video}

Given a video, suppose the video frame at time $t$ has been processed to form an attributed graph $G^{(t)}=(V^{(t)}, p^{(t)}, g^{(t)}, e^{(t)})$ using object detection techniques (\emph{e.g.}, CNN-based, or as processed in the Stanford Drone Dataset), where
\begin{compactitem}
\item $V^{(t)}$ is the set of nodes (\emph{i.e.}, moving objects \cite{yazdi2018new});
\item $p^{(t)}(v) = [p^{(t)}_{x} (v), p^{(t)}_{y} (v)]: V^{(t)} \rightarrow \mathbb{R}^{2}$ gives two position values (\emph{i.e.}, horizontal and vertical positions) of the center of object $v$ in the frame;
\item $g^{(t)}(v): V^{(t)} \rightarrow \mathbb{R}^{k}$ gives $k$ features of object $v$ such as red/green/blue values on the left and right, at the top and bottom, and in the center;
\item $e^{(t)}(u,v): V^{(t)} \times V^{(t)} \rightarrow \mathbb{R}$ gives the weight of the link between nodes $u$ and $v$  -- the weight can be the Euclidean distance between the two nodes on the frame:
	\begin{equation}
		e^{(t)}(u,v) = {\| p^{(t)}(u) - p^{(t)}(v)  \|}^2, u, v \in V^{(t)}.
	\end{equation}
\end{compactitem}

Given a video with frames at $1 \dots T$, we have attributed graph sequence $G^{(t)} {|}^{T}_{t=1}$. We denote the set of nodes in the graph sequence by $V = {\cup}^T_{t=1} V^{(t)}$. The goal of our approach is to learn the representations of each node (called \emph{node embeddings}) in each graph that preserve spatial information in the graph and dynamic information of the node through past graphs: $f(v,t): V \times \{1 \dots T\} \rightarrow \mathbb{R}^{d}$, denoted by $\mathbf{v}^{(t)}$, where $d$ is the number of dimensions of node embeddings. The node embeddings can be used for tasks such as object tracking, forecasting, and malicious behavior detection.

Our proposed GNN model has two parts: One is an algorithm for node embedding generation given raw data and model parameters; the other are loss function(s) based on self-supervised task(s) for training the model parameters. 

\begin{figure}[t]
	\centering
	\includegraphics[width=0.43\textwidth]{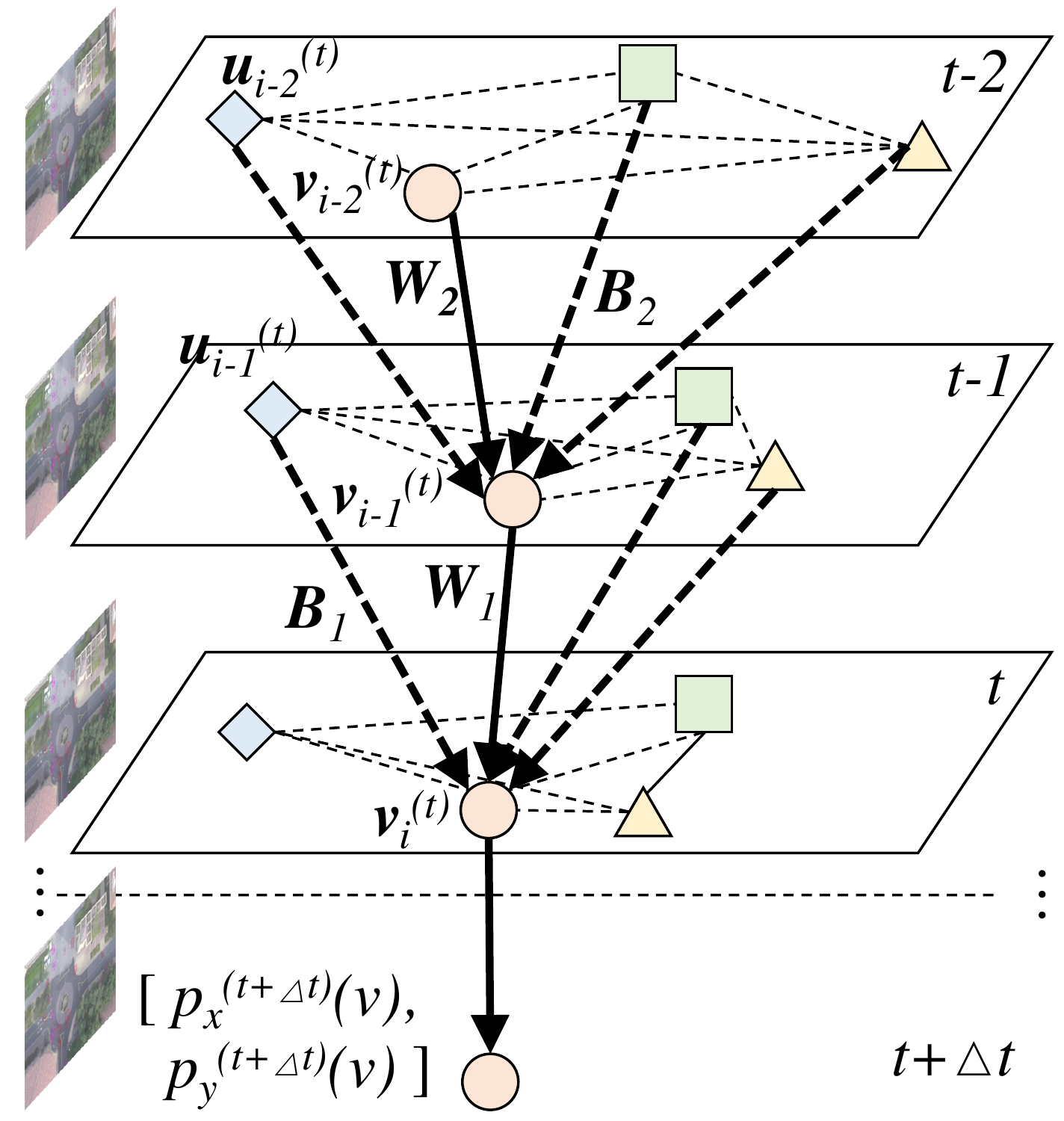}
	\hspace{0.1in}
	\includegraphics[width=0.43\textwidth]{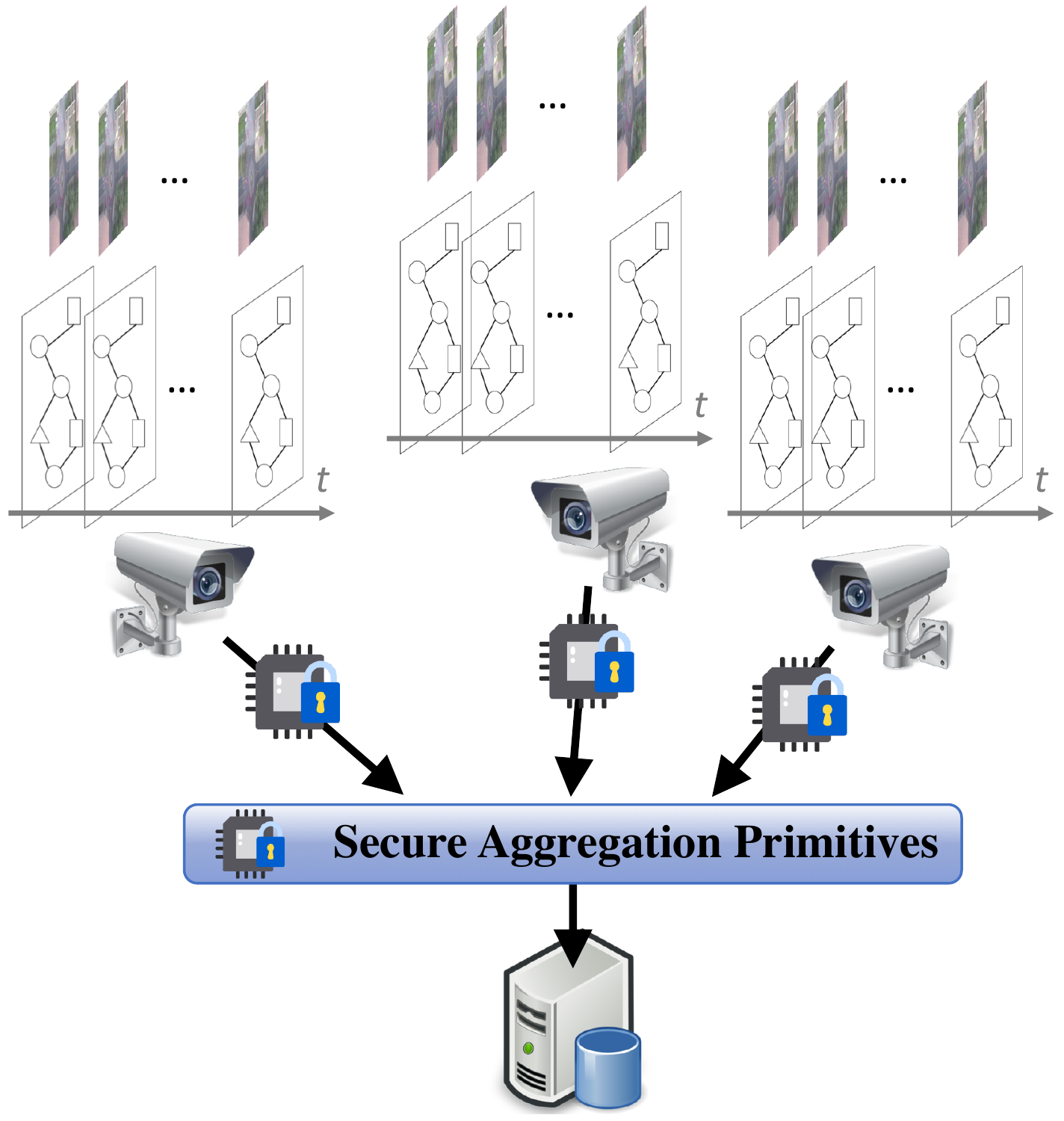}
	\caption{Dynamic GNN aggregates spatial information within graphs and temporal information across neighboring graphs. It is trained by the MSE loss on the task of predicting future positions of objects. (Left) We protect federated averaging functions with novel secure aggregation primitives to prevent inference attack by malicious parameter servers. (Right)}
	\label{fig:framework}
\end{figure}

\textbf{Node embedding generation:} First, we use matrix $\mathbf{M} \in \mathbb{R}^{d \times (k+2)}$ to transform from node's raw features $[p^{(t)}(v),g^{(t)}(v)]$ to the initial latent embeddings $\mathbf{v}^{(t)}_{0} \in \mathbb{R}^{d}$:
	\begin{equation}
		\mathbf{v}^{(t)}_{0} = \sigma \left(\mathbf{M} \cdot [p^{(t)}(v),g^{(t)}(v)] \right),
	\end{equation}
where $\sigma(\cdot)$ is an activation function, which can be sigmoid, hyperbolic tangent, ReLU, etc.

Second, for the $i$-th layer of the neural network ($i \in \{1 \dots n\}$, where $n$ is the number of layers), which means in the $i$-th iteration of the embedding generation algorithm, we generate the embedding vector $\mathbf{v}^{(t)}_{i} \in \mathbb{R}^{d}$ for node $v$ in graph $G^{(t)}$ as follows:
	\begin{equation}
		\mathbf{v}^{(t)}_{i} = \sigma \left( \alpha \cdot \mathbf{B}_{i} \cdot \textsc{Agg}_{u \in V^{(t-1)} \setminus \{v\}} \pi(e^{(t-1)}(u,v)) \cdot \mathbf{u}^{(t-1)}_{i-1} + \beta \cdot \mathbf{W}_{i} \cdot \mathbf{v}^{(t-1)}_{i-1} + (1 - \alpha - \beta) \cdot \mathbf{v}^{(t)}_{i-1} \right),
	\end{equation}
where
(1) $\mathbf{u}^{(t-1)}_{i-1}$ is the embedding vector of node $u$ as a neighboring node of $v$ in graph $G^{(t-1)}$ at the $(i-1)$-th iteration; $\mathbf{v}^{(t-1)}_{i-1}$ is the embedding vector of node $v$ in graph $G^{(t-1)}$ at the $(i-1)$-th iteration; $\mathbf{v}^{(t)}_{i-1}$ is the embedding vector of node $v$ in graph $G^{(t)}$ at the $(i-1)$-th iteration;
(2) $\mathbf{B}_{i} \in \mathbb{R}^{d \times d}$ is the transformation matrix from the aggregated information of neighboring nodes on the $(i-1)$-th layer to the $i$-th layer; $\mathbf{W}_{i} \in \mathbb{R}^{d \times d}$ is the transformation matrix from the embedding of a node on the $(i-1)$-th layer to the $i$-th layer;
(3) $\alpha$ is a hyperparameter weighting the aggregation of neighboring nodes in the previous graph; $\beta$ is a hyperparameter weighting the node in the previous graph;
(4) $\textsc{Agg}$ is an aggregation function, which can be mean pooling, max pooling, or LSTM aggregator,
and (5) $\pi(e)$ defines the importance of aggregating from a neighboring node: A shorter distance (i.e., weight on the link between the two nodes) indicates a higher importance, so one choice is $\pi(e) = \frac{1}{e}$.

The final embedding vectors are $\{ \mathbf{\hat{v}}^{(t)} \equiv \mathbf{v}^{(t)}_{n} \} {|}^{T}_{t=1}$, generated with model parameters including $\mathbf{M}$, $\mathbf{B}_{i} {|}^{n}_{i=1}$, and $\mathbf{W}_{i} {|}^{n}_{i=1}$ from graph sequence data $G^{(t)} {|}^{T}_{t=1}$. Each final embedding vector $\mathbf{\hat{v}}^{(t)}$ preserves structural and dynamic information of node $v$ in $G^{\min\{t-n,1\}} \dots, G^{(t-1)}, G^{(t)}$.

\textbf{Self-supervised loss:} We expect the final embedding vector $\mathbf{\hat{v}}^{(t)}$ can be predictive for future values of node's positions which can be denoted by $[p^{(t+\Delta t)}_{x} (v), p^{(t+\Delta t)}_{y} (v)]$, i.e., $v$'s positions after $\Delta t$ frames. So we have the loss below by introducing $\mathbf{A} \in \mathbb{R}^{2 \times d}$ from latent space to position space:
	\begin{equation}\label{eq:loss}
		\mathcal{L} (\mathbf{M}, \mathbf{B}_{i} {|}^{n}_{i=1}, \mathbf{W}_{i} {|}^{n}_{i=1}, \mathbf{A}) = \sum_{t=2}^{T-\Delta t} \sum_{v \in V^{(t)}} {\| p^{(t+\Delta t)}(v) - \mathbf{A} \cdot \mathbf{\hat{v}}^{(t)} \|}^{2}.
	\end{equation}

\textit{Complexity analysis:} The time complexity of the Dynamic GNN is $O(r^n d^2 \sum_{t=1}^T {|V^{(t)}|})$, where $n$ is the number of layers, $r \leq |V^{(t)}|$ is the number of spatial neighbors for each node, and $d$ is the number of dimensions of node embeddings. The memory complexity is $O(r^n d + n d^2)$. Usually, the number of objects $|V^{(t)}|$ in a single graph is not too big (often between 2 and 10). And $n$ is usually 1, 2, or 3. When $|V^{(t)}|$ turns to be too big, we can sample $r$ spatial neighbors for the training \cite{hamilton2017inductive}.

\section{Federated Learning for Training Distributed Dynamic GNNs}

\textbf{Joint optimization with multi-user graph sequences from distributed cameras.} Given $m$ videos, which can be represented as $G^{(j,t)} {|}^{j=m,t=T_j}_{j=1,t=1}$, where $T_j$ is the number of graphs (\emph{i.e.}, video frames) in the $j$-th graph sequence (\emph{i.e.}, video), we extend the dynamic GNN algorithm presented above to generate final embedding vectors $ \mathbf{\hat{v}}^{(j,t)} {|}^{j=m,t=T_j}_{j=1,t=1}$ and extend the self-supervised loss as follows:
\begin{equation}\label{eq:selfsupervised-loss}
\mathcal{L} = \sum_{j=1}^{m} \sum_{t=2}^{T_j-\Delta t} \sum_{v \in V^{(j,t)}} {\| p^{(j,t+\Delta t)}(v) - \mathbf{A} \cdot \mathbf{\hat{v}}^{(j,t)} \|}^{2}.
\end{equation}


\textbf{Federated optimization.} Privacy, security, and scalability have become critical concerns in distributed surveillance \cite{bonawitz2019towards}. An approach that has the potential to address these problems is federated learning \cite{konevcny2016federated}. Federated learning (FL) is an emerging decentralized privacy-protection training technology that enables clients to learn a shared global model without uploading their private local data to a central server. In each training round, a local device downloads a shared model from the central server, trains the downloaded model over the individuals' local data and then sends the updated weights or gradients back to the server. On the server, the uploaded models from the clients are aggregated to obtain a new global model \cite{zhu2019multi,yang2019federated}. Federated dynamic GNN aims to minimize the loss function in \Cref{eq:selfsupervised-loss} but in a distributed scheme:
\begin{equation}
\min_{\Theta} \mathcal{L}(\Theta) = \sum_{j=1}^{m} \frac{N_j}{N} \mathcal{L}_j (\Theta),~\text{where}~\mathcal{L}^{(j)} (\Theta) = \frac{1}{N_j} \sum_{v \in V^{(j,:)}} \mathcal{L}_v (\Theta),
\end{equation}
where the size of data indexes $N_j = \sum_{t} |V^{(j,t)}|$ and $N = \sum_{j} N_j$, $j$ is the index of $m$ clients, $\mathcal{L}^{(j)} (\Theta)$ is the loss function of the $j$-th local client, and $\Theta$ represents neural parameters of the dynamic GNN model. Optimizing the loss function $\mathcal{L}(\Theta)$ in FL is equivalent to minimizing the weighted average of local loss function $\mathcal{L}(\Theta)$.

Each user performs local training to calculate individual loss $\mathcal{L}^{(j)} (\mathbf{M}, \mathbf{B}_{i} {|}^{n}_{i=1}, \mathbf{W}_{i} {|}^{n}_{i=1}, \mathbf{A})$ (for the $j$-th user), after the loss gradient $\nabla \mathcal{L}^{(j)} (\mathbf{M}, \mathbf{B}_{i} {|}^{n}_{i=1}, \mathbf{W}_{i} {|}^{n}_{i=1}, \mathbf{A})$ is calculated. Each user submits this individual gradient to a central \textit{parameter server} who takes the following gradient descent step:
\begin{equation}\label{eq:federated-GD}
\Theta^{t+1}\leftarrow \Theta^{t}-\eta \cdot {\textsc{FedAvg}}_{j=1 \dots m} \left( \alpha_j \nabla \mathcal{L}^{(j)} (\mathbf{M}, \mathbf{B}_{i} {|}^{n}_{i=1}, \mathbf{W}_{i} {|}^{n}_{i=1}, \mathbf{A}) \right),
\end{equation}
where $t$ is for iteration (not time or graph index), $\Theta$ is for neural network parameters, $\eta$ is the learning rate, and $\alpha_j$ is the weight of the $j$-th user's individual gradient in the weighted average $\textsc{FedAvg}$. $\alpha_j$ is usually proportional to the size of the user's dataset. The algorithm based on $\textsc{FedAvg}$ can effectively reduce communication rounds by simultaneously increasing local training epochs and decreasing local mini-batch sizes \cite{mcmahan2016communication,zhu2019multi}.

\section{Privacy-preserving FL with Improved Secure Aggregation}\label{sec:secure-aggregation}

Regular FL without explicit security mechanisms has been shown to be vulnerable to various inference attacks \cite{melis2019exploiting,nasr2019comprehensive,truex2019demystifying,abadi2016deep,fredrikson2015model}. Namely, adversaries who observe the updates of individual models become able to infer significant information about the individual training datasets (e.g., distribution of the training datasets or even samples or training datasets). When the parameter server in the FL is compromised, individual data (e.g., raw videos or the features) is compromised as well with the inference attacks.

A secure aggregation scheme allows a group of distrustful users $\nu\in U$ ($U$ is the set of $m$ users) with private input $x_\nu$ to compute an aggregate value ${\textsc{FedAvg}}_{\nu\in U}x_\nu$ such as $\sum_{\nu\in U}x_\nu$without disclosing individual $x_\nu$'s to others. Similar to existing work \cite{bonawitz2016practical,bonawitz2017practical,bonawitz2019towards}, we leverage  secure aggregation to thwart attacks. However, we take advantage of different approaches and mitigate their shortcomings. 

\textbf{Secure aggregation with pair-wise one-time pads.} In the secure aggregation adopted by Bonawitz \textit{et al.} \cite{bonawitz2016practical,bonawitz2017practical,bonawitz2019towards}, a user $\nu$ chooses a random number $s_{\nu,\mu}\in \mathbb{Z}_q$ for every other user $\mu$, where $\mathbb{Z}_q$ is a set of integers $\{0,1,\cdots,q-1\}$. Specially, $s_{\nu,\mu}=0$ if $\nu=\mu$. Then, all pairs of users $\nu$ and $\mu$ exchange $s_{\nu,\mu}$ and $s_{\mu,\nu}$ over secure communication channels and compute the one-time pads as $p_{\nu,\mu}=s_{\nu,\mu}-s_{\mu,\nu}$. Then, each user masks $x_\nu$ as $y_\nu=x_\nu+\sum_{\mu\in U}p_{\nu,\mu}\mod q$. Every user $\nu$ sends $y_{\nu}$ to the server who computes the sum over $y_{\nu}$. Then, it follows that
\begin{equation}
    \sum_{\nu\in U}y_\nu = \sum_{\nu\in U}x_\nu+\sum_{\nu\in U}\sum_{\mu\in U} p_{\nu,\mu}= \sum_{\nu\in U}x_\nu +\sum_{\nu\in U}\sum_{\mu\in U}s_{\nu,\mu}-\sum_{\nu\in U}\sum_{\mu\in U}s_{\mu,\nu}=\sum_{\nu\in U} x_\nu~~ (\text{mod}~q)
\end{equation}

Such masking with one-time padding guarantees perfect secrecy (i.e., no information about $x_\nu$ is revealed from $y_\nu$) as long as the bit length of $q$ is larger than the bit length of $x_\nu$. Such secure aggregation requires that all users frequently share one-time pads at every aggregation,
because the one-time pads cannot be reused. This leads to high communication overhead among the users. The benefit of such a scheme is that it does not rely on a trusted key dealer unlike the following scheme.

\textbf{Secure aggregation of time-series data.} In the secure aggregation for time-series data \cite{shi2011privacy,joye2013scalable}, a trusted key dealer randomly samples for each user $\nu\in U$ their one-time pads $\{p_\nu\}{|}_{\nu\in U}$ from $\mathbb{Z}$ such that $\sum_{\nu\in U}p_\nu=0$. This can be done trivially by choosing the first $|U|-1$ numbers randomly and let the last number be the negative sum of the $|U|-1$ numbers (note that $-x$ is equal to $q-x$ modulo $q$ for any $x\in\mathbb{Z}_q$). Then, each $p_\nu$ is securely distributed to each user via secure communication channels. Each user $\nu$ then  masks $x_\nu$ as $y_\nu=(1+ N)^{x_\nu}H(t)^{p_\nu}\mod N^2$, where $H:T\rightarrow \mathbb{Z}_{N^2}$ is a cryptographic hash function, $t$ is the time at which the aggregation needs to be performed, and $N$ is the product of two distinct unknown prime numbers (i.e., RSA number). Then, it follows that:
\begin{equation}
    (\prod_{\nu\in U}y_\nu\big)=\prod_{\nu\in U}H(t)^{p_\nu}\cdot \prod_{\nu\in U}(1+ N)^{x_\nu}=(1+N)^{\sum_{\nu\in U}x_\nu}=(1+N\sum_{\nu\in U}x_\nu)        ~~(\text{mod}~N^2),
\end{equation}
where the last equality holds due to the binomial theorem, i.e., $(1+N)^x=1+xN$ (mod $N^2$). Then, it further follows that:
\begin{equation}
    \frac{\big((\prod_{\nu\in U}y_\nu)-1\big)\mod N^2}{N}\mod N=\sum_{\nu\in U}x_\nu.
\end{equation}

Such masking guarantees semantic security (i.e., no statistical information about $x_\nu$ is disclosed from $y_\nu$) as long as $N$ is sufficiently large and the Decision Composite Residuosity (DCR) problem \cite{joye2013scalable} becomes hard. Such secure aggregation requires a trusted key dealer who performs the computation, however such an entity is hard to find in real-life applications. Even if there is one, it becomes the single-point-of-failure whose compromise leads to the compromise of the whole system. The benefit of such scheme is that it does not require frequent key sharing because $H(t)^{p_\nu}$ is computationally indistinguishable from a random number from $\mathbb{Z}_{N^2}$ as long as $p_\nu$ is kept secret, and users can re-use $p_\nu$ over and over as long as no same $t$ is used in the aggregation.

\textbf{Our improved secure aggregation.} We present our secure aggregation scheme that takes the best of the both worlds. Namely, we combine the two types of secure aggregation schemes above to let users re-use the shared pads/keys without extra sharing or trusted key dealers.

We first let all users $\{\nu\}|_{\nu\in U}$ generate the one-time pads $p_{\nu,\mu}$ in the same way as in the secure aggregation with pair-wise one-time pads. Then, we let every user calculate $p_\nu=\sum_{\mu\in U}p_{\nu,\mu}$, whose sum $\sum_{\nu\in U}p_\nu$ is equal to 0 for some modulus. Then, we let all users use these $\{p_\nu\}|_{\nu\in U}$ to participate in the secure aggregation of time-series data. By doing so, users can use $H(t)$'s and $p_\nu$'s to mask their inputs without relying on a trusted key dealer. At the same time, they can re-use the same pads $p_\nu$'s repeatedly as long as $H(t)$ is different every time. Informally, such a masking guarantees the correct aggregation at the parameter server side, and it also guarantees that adversaries cannot infer any  information related to individual users' input data other than the length of the masked data. We present formal definitions and proofs in the appendix.

\textbf{Mapping between real numbers and integers.} All the computations in our secure aggregation scheme are integer computations, and we need to use integers to represent real numbers. We leverage the fixed point representation \cite{jung2016pda}. Given a real number $x$, its fixed point representation is given by $[x]=\lfloor x\cdot 2^e \rceil$ for a fixed integer $e$. Then, it follows that $[x\pm y]=[x]\pm[y]$. With such  homomorphism, users can convert a real number $x$ to its integer version $[x]$ and participate in the secure aggregation. The third-party aggregator can compute the sum $\sum_{\nu\in U} [x_\nu]$ which is equal to $[\sum_{\nu\in U}x_\nu]$, and one can approximately compute $\sum_{\nu\in U}x_\nu$ by computing the following division: $\sum_{\nu\in U}x_\nu \approx \frac{[\sum_{\nu\in U}x_\nu]}{2^e}$.
The approximation error is bounded above by $2^{-(e+1)}|U|$.

\section{Experiments}

\subsection{Experimental settings}

\begin{table}[h]
    \centering
    \begin{tabular}{l||c|c|c|c||c|c|c|c|c|c}
    \toprule
        Scene & \#Edges & Train & Valid. & Test & Bic. & Ped. & Skate. & Cart & Car & Bus \\
        \midrule
        bookstore & 2,112 & 578 & 298 & 869 & 32.89 & 63.94 & 1.63 & 0.34 & 0.83 & 0.37 \\ \hline
        coupa & 4,901 & 553 & 717 & 2,002 & 18.89 & 80.61 & 0.17 & 0.17 & 0.17 & 0 \\ \hline
        hyang & 3,076 & 822 & 199 & 1,813 & 27.68 & 70.01 & 1.29 & 0.43 & 0.50 & 0.09 \\ \hline
        little & 4,413 & 485 & 602 & 1,519 & 56.04 & 42.46 & 0.67 & 0 & 0.17 & 0.67 \\
    \bottomrule
    \end{tabular}
    \vspace{0.05in}
    \caption{Statistics of four video datasets (graph sequences). Objects in the videos have various types (in \%): Bic. = Bicyclist, Ped. = Pedestrian, Skate. = Skateboarder.}
    \label{tab:datasets}
\end{table}

\textbf{Graph sequence data.} We transform videos in the Stanford Drone Dataset into graph sequences. Each 30-fps video spans over 10,000 frames (\emph{i.e.}, over 333 seconds). We use the first minute for training, last 3 minutes for testing, and the last 4$^{th}$ for validation. Statistics can be found in Table~\ref{tab:datasets}.

\textbf{Parameter settings.} The number of object's raw features is $k+2$: $2$ is for the object box's horizontal and vertical positions. $k=17$ is for box width, box height, and RGB colors at the center, left, right, top, and bottom of the box. The hyperparameters $\alpha$ is set as $0.1$, $\beta$ is set as $0.1$ and the number of layers is set as $2$ for the best performance. Here $\alpha$ is the weight for aggregating spatial information from neighbors; $\beta$ is the weight for aggregating dynamic information from neighboring graphs. The aggregation is applied every 10 epochs.

We will set the number of dimensions of node embeddings $d$ a value in $\{32, 64, 128, 256, 512\}$. We will simulate federated learning with the number of users $m$ in $\{1, 2, 5, 10\}$. $m=1$ means we disable federated optimization but have all data on a single user.

\textbf{Computational resource.} Our machine is an iMac with 4.2 GHz Quad-Core Intel Core i7, 32GB 2400 MHz DDR4 memory, and Radeon Pro 580 8GB Graphics.

\subsection{Results on effectiveness}

Figure~\ref{fig:results} presents experimental results in four different scenes. We use Root Mean Square Error (RMSE) to evaluate the performance of predicting horizontal position $p_x$ and vertical position $p_y$. We observe, first, the RMSE of Federated Dynamic GNN (Feddy) for any number of users and any number of dimensions is smaller than the best method (\emph{i.e.}, MLP) on learning raw features. Second, Feddy performs better when the number of dimensions $d$ is bigger. When $d \geq 128$, the performances under different number of users ($1$, $2$, $5$, or $10$) show very small difference. It means federated optimization achieves a consistent global model.

\begin{figure}[t]
    \centering
    \includegraphics[height=3.42cm]{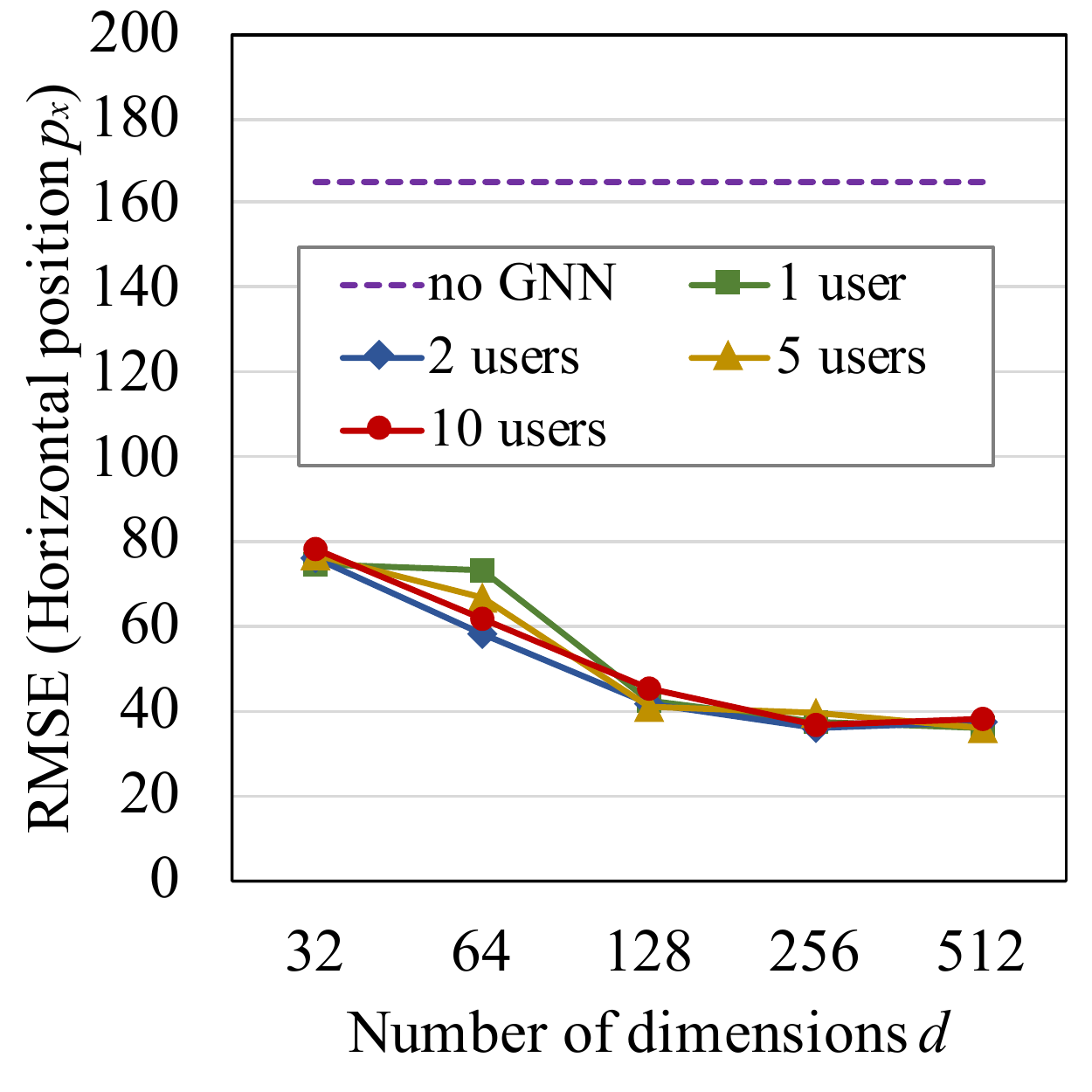}
    \includegraphics[height=3.42cm]{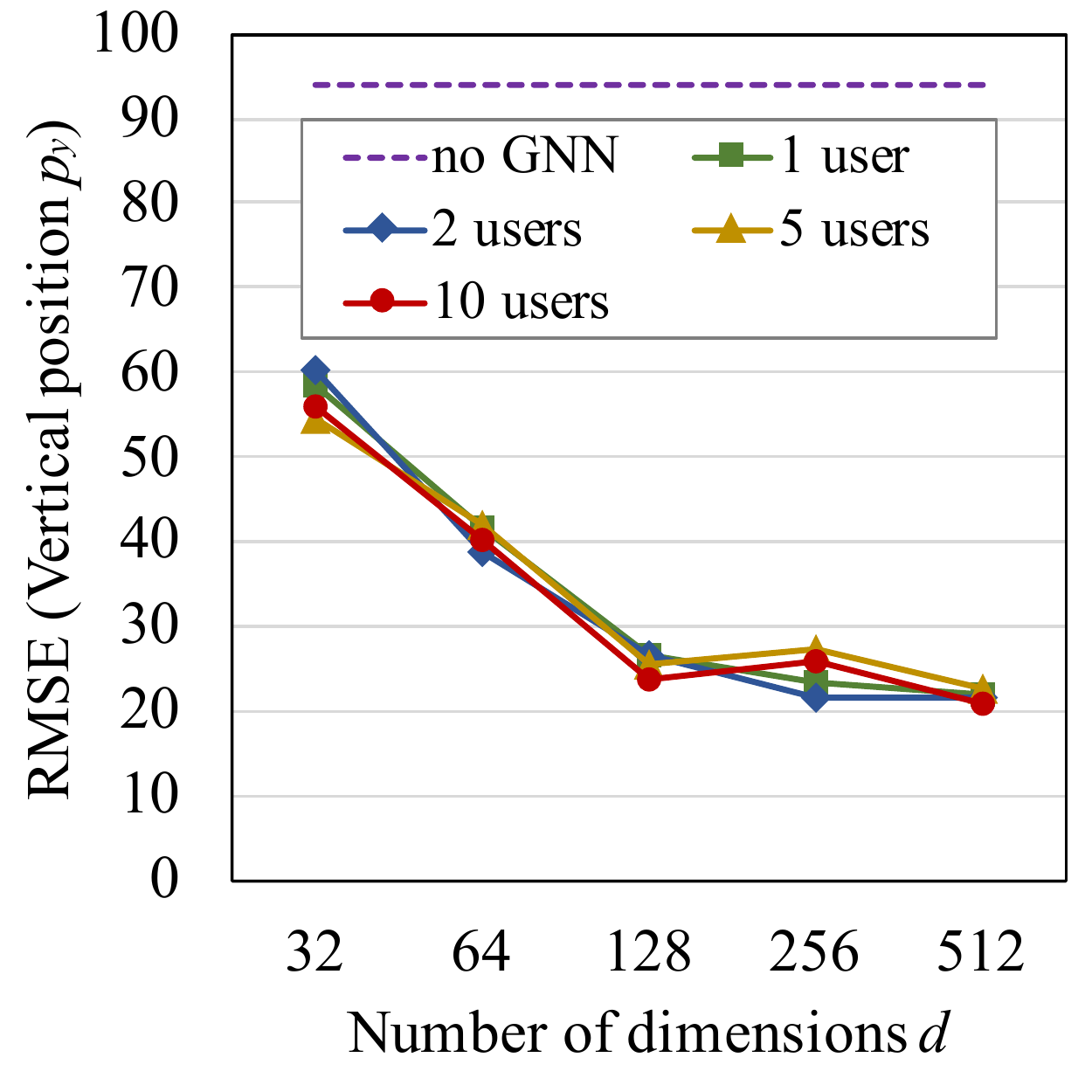}
    \vline
    \includegraphics[height=3.42cm]{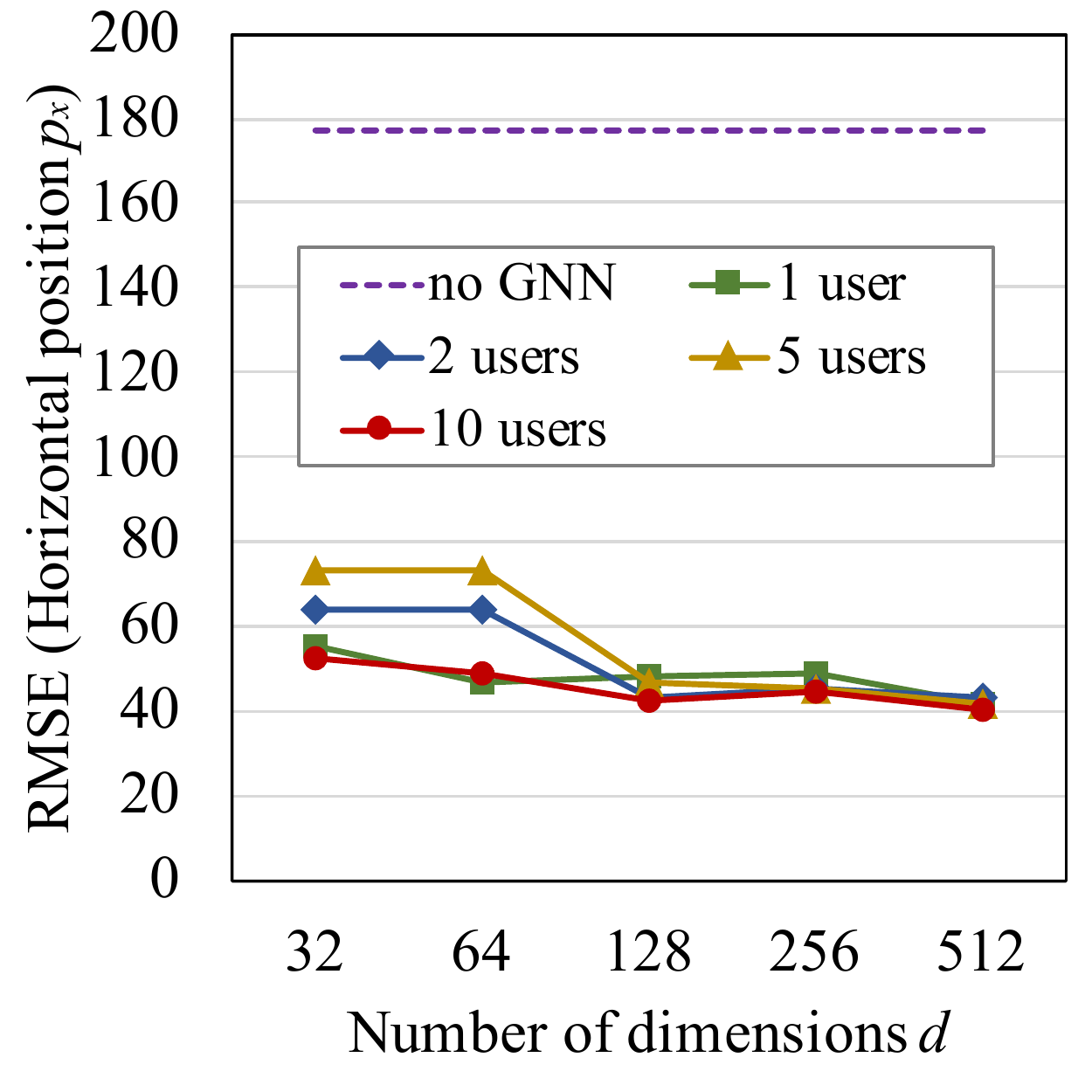}
    \includegraphics[height=3.42cm]{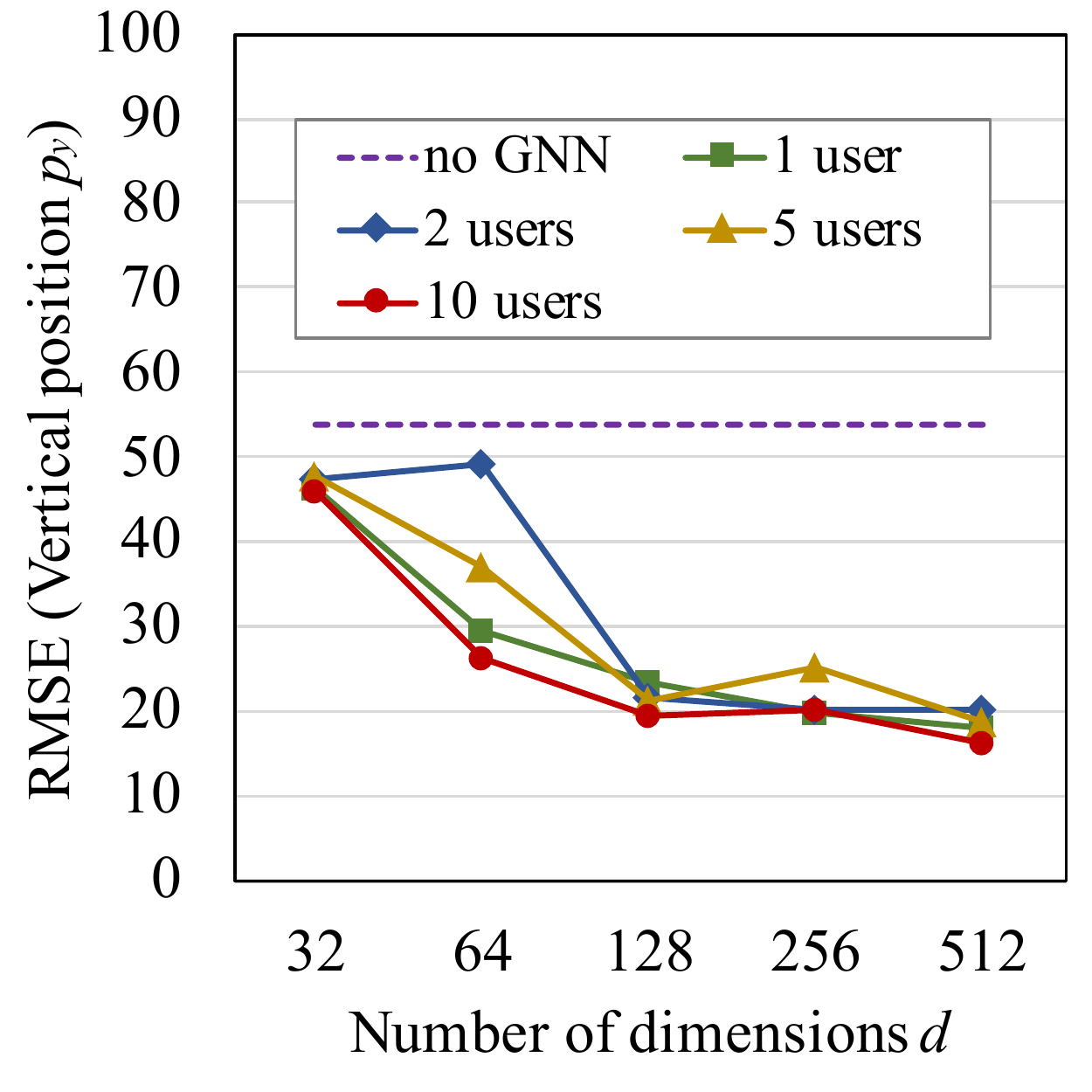} \\
    \includegraphics[height=3.42cm]{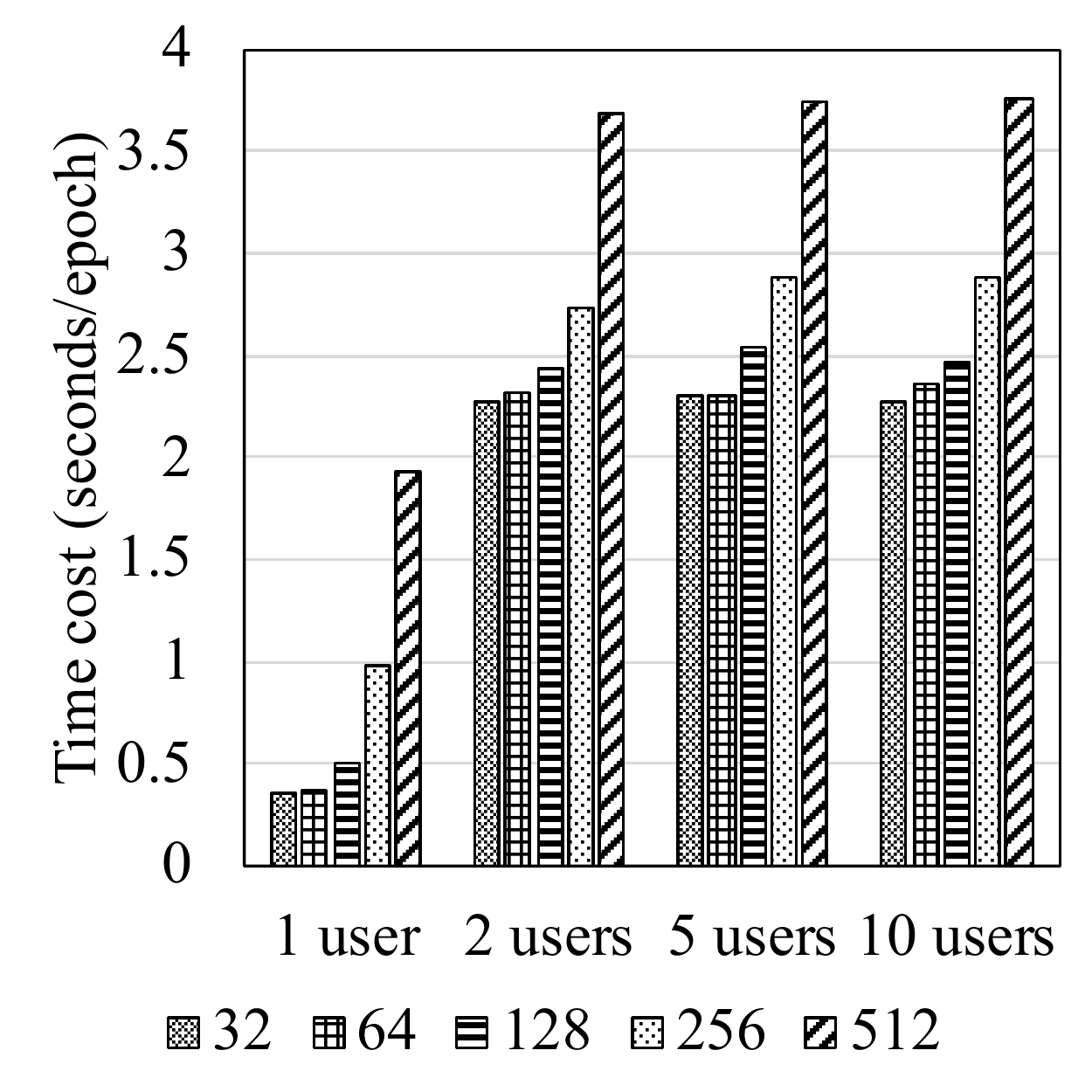}
    \includegraphics[width=3.42cm]{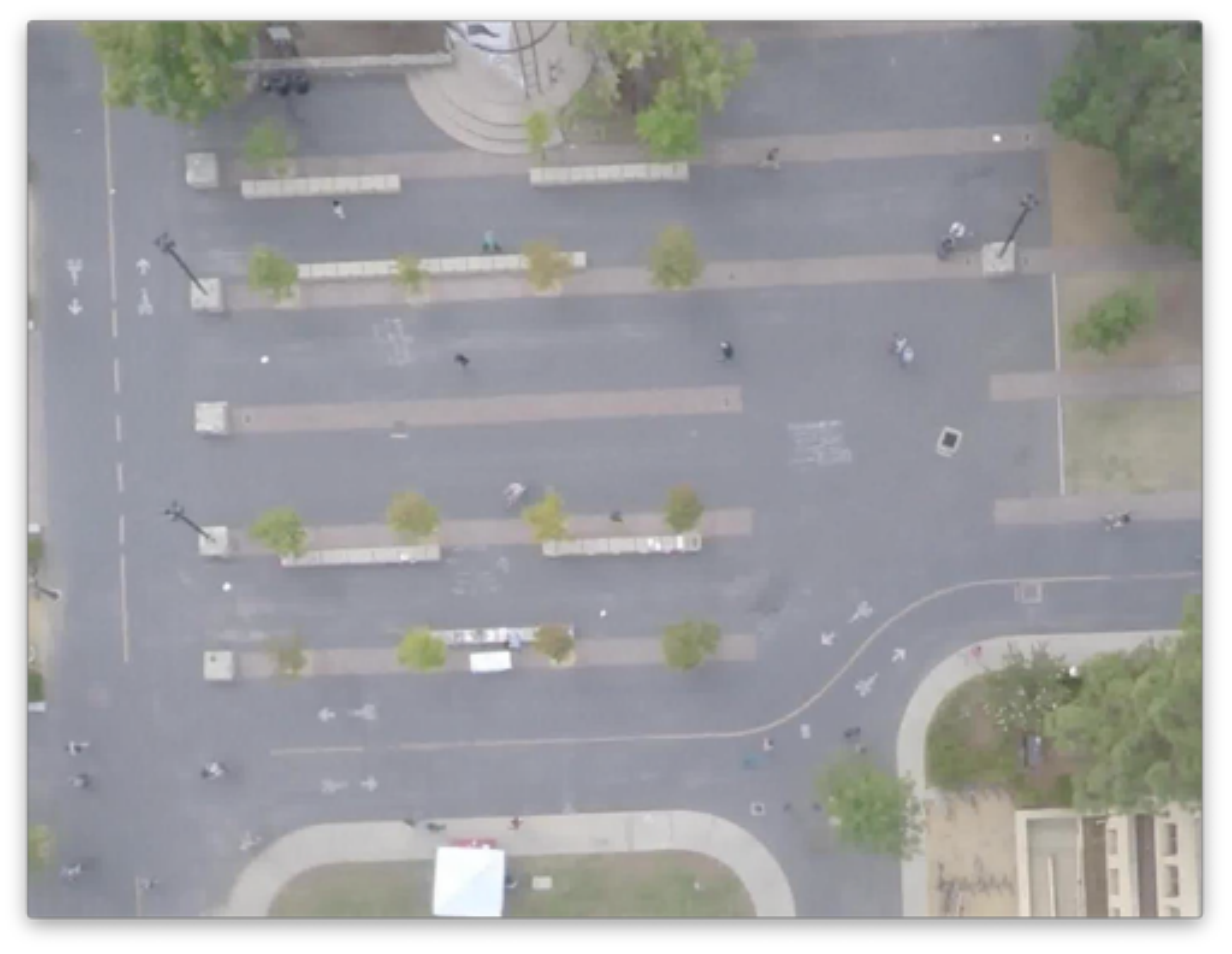}
    \vline
    \includegraphics[height=3.42cm]{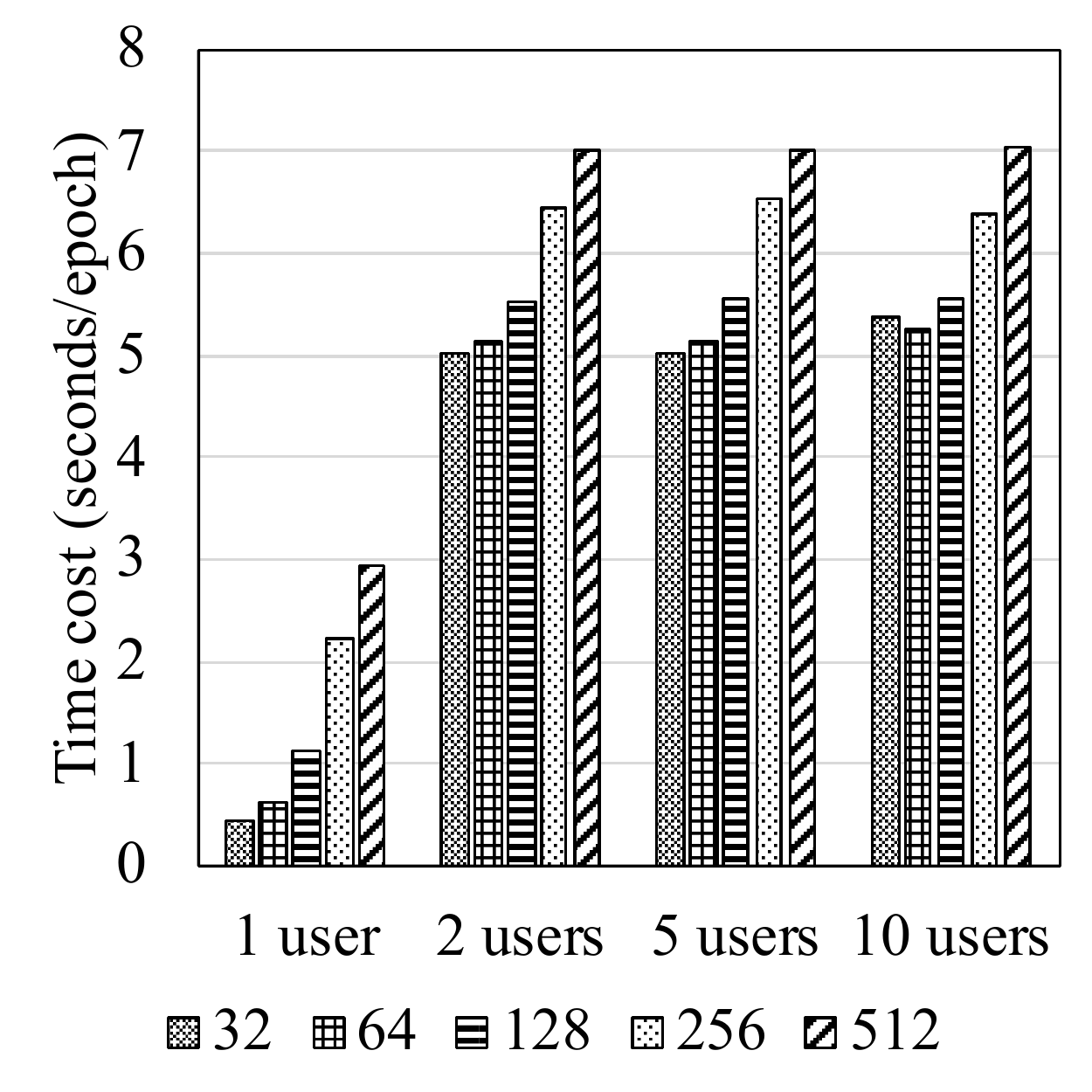}
    \includegraphics[width=3.42cm]{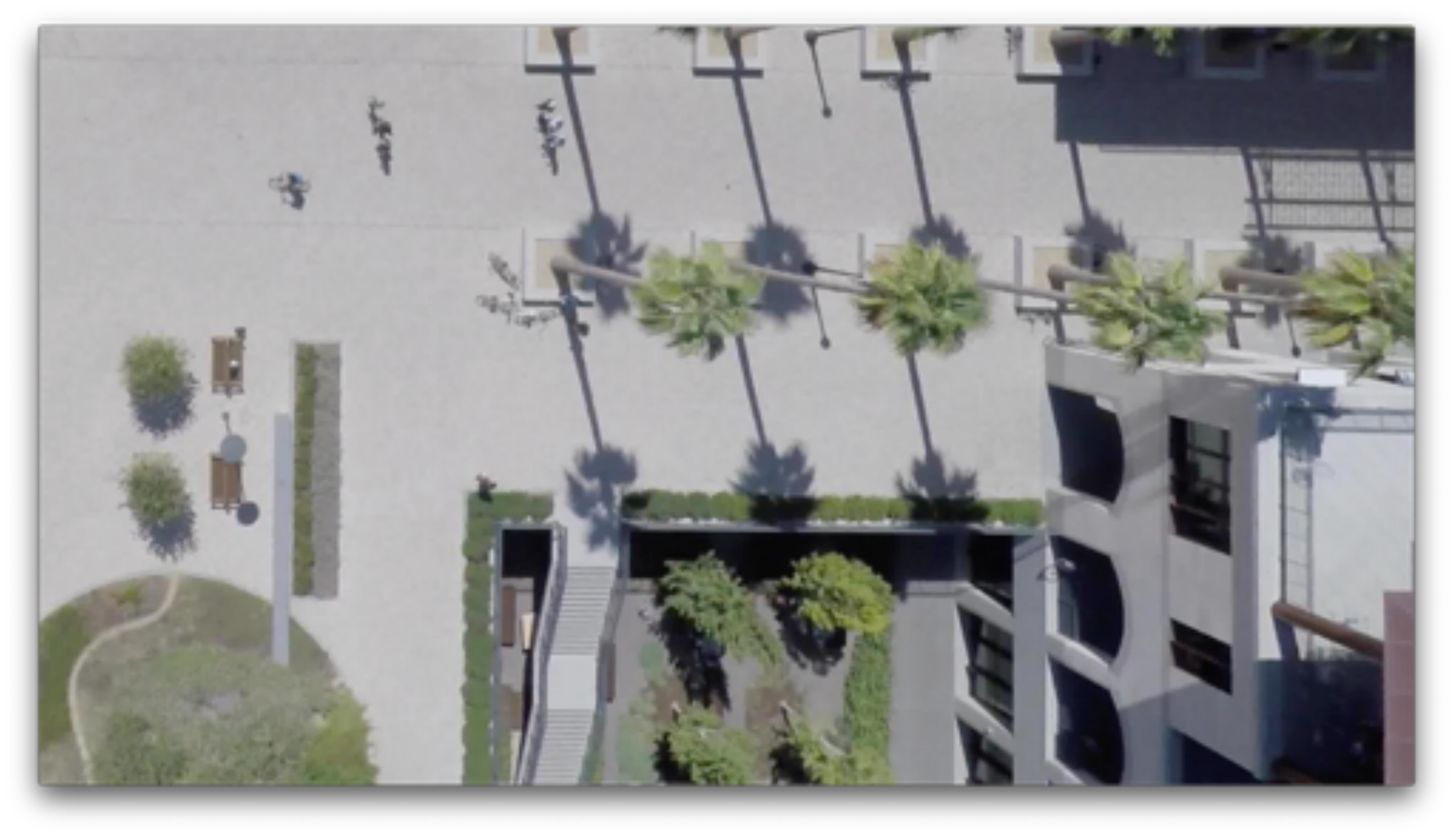} \\
    (a) Stanford Bookstore (bookstore): 1424$\times$1088
    \quad \quad
    (b) Stanford Coupa Cafe (coupa): 1980$\times$1093 \\
    \includegraphics[height=3.42cm]{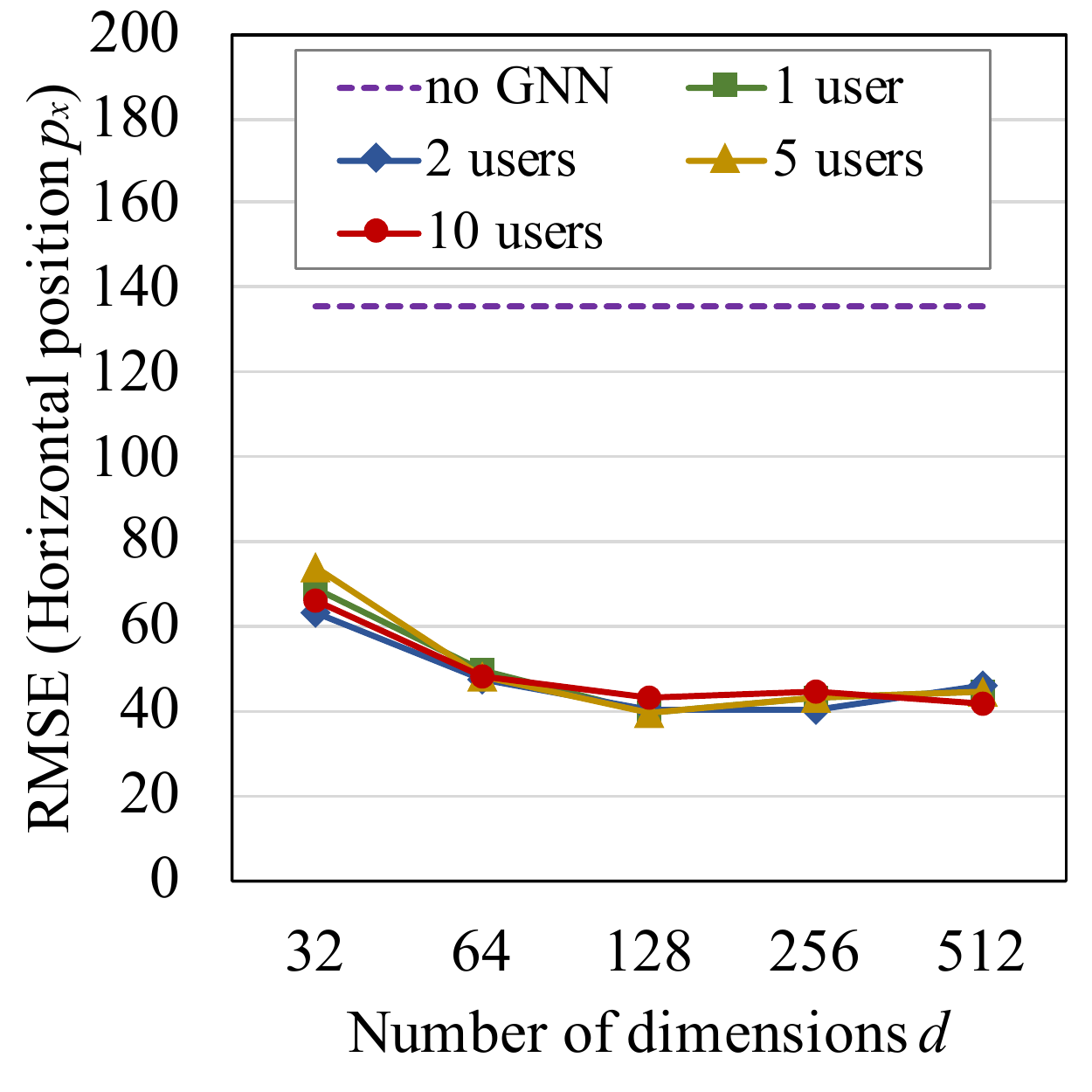}
    \includegraphics[height=3.42cm]{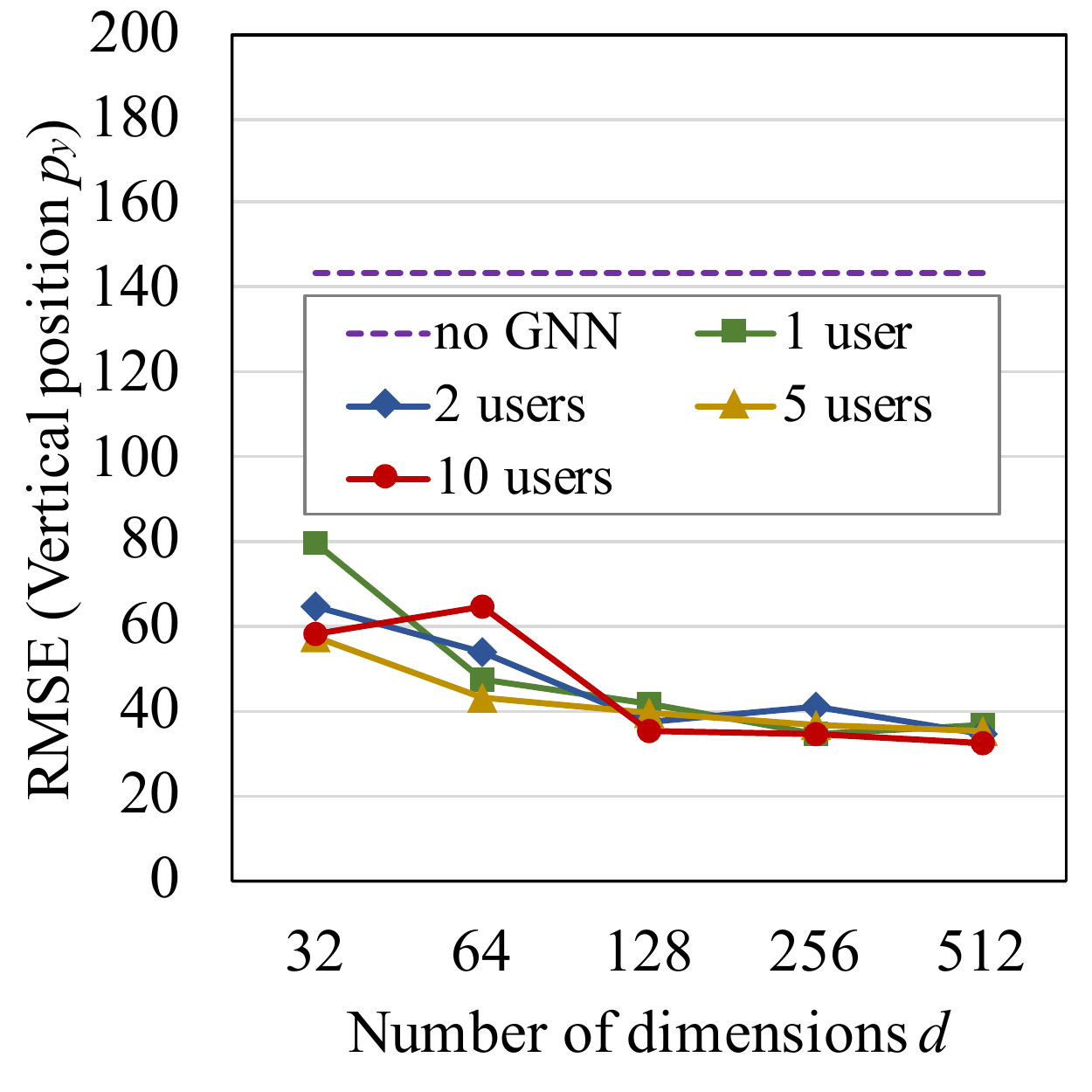}
    \vline
    \includegraphics[height=3.42cm]{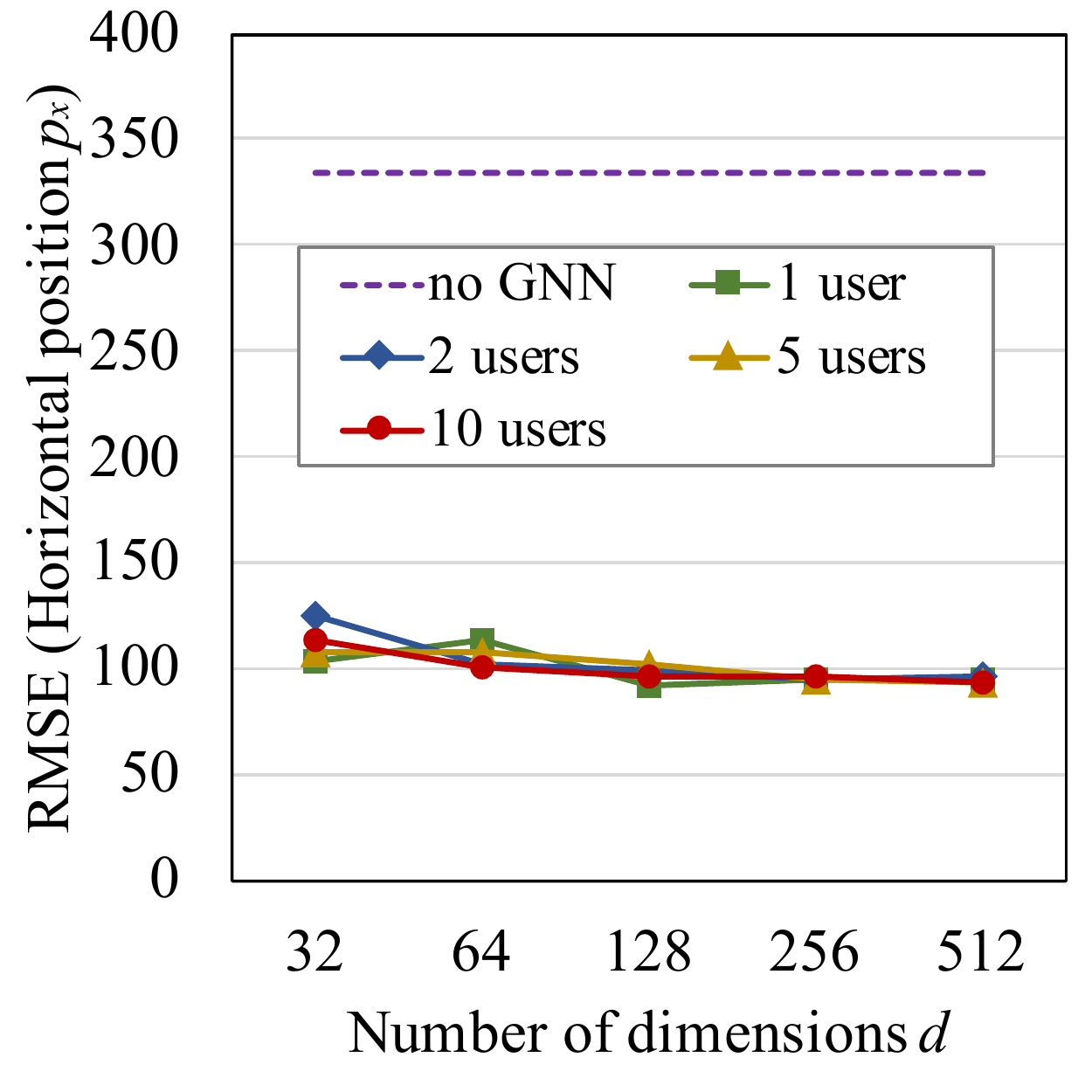}
    \includegraphics[height=3.42cm]{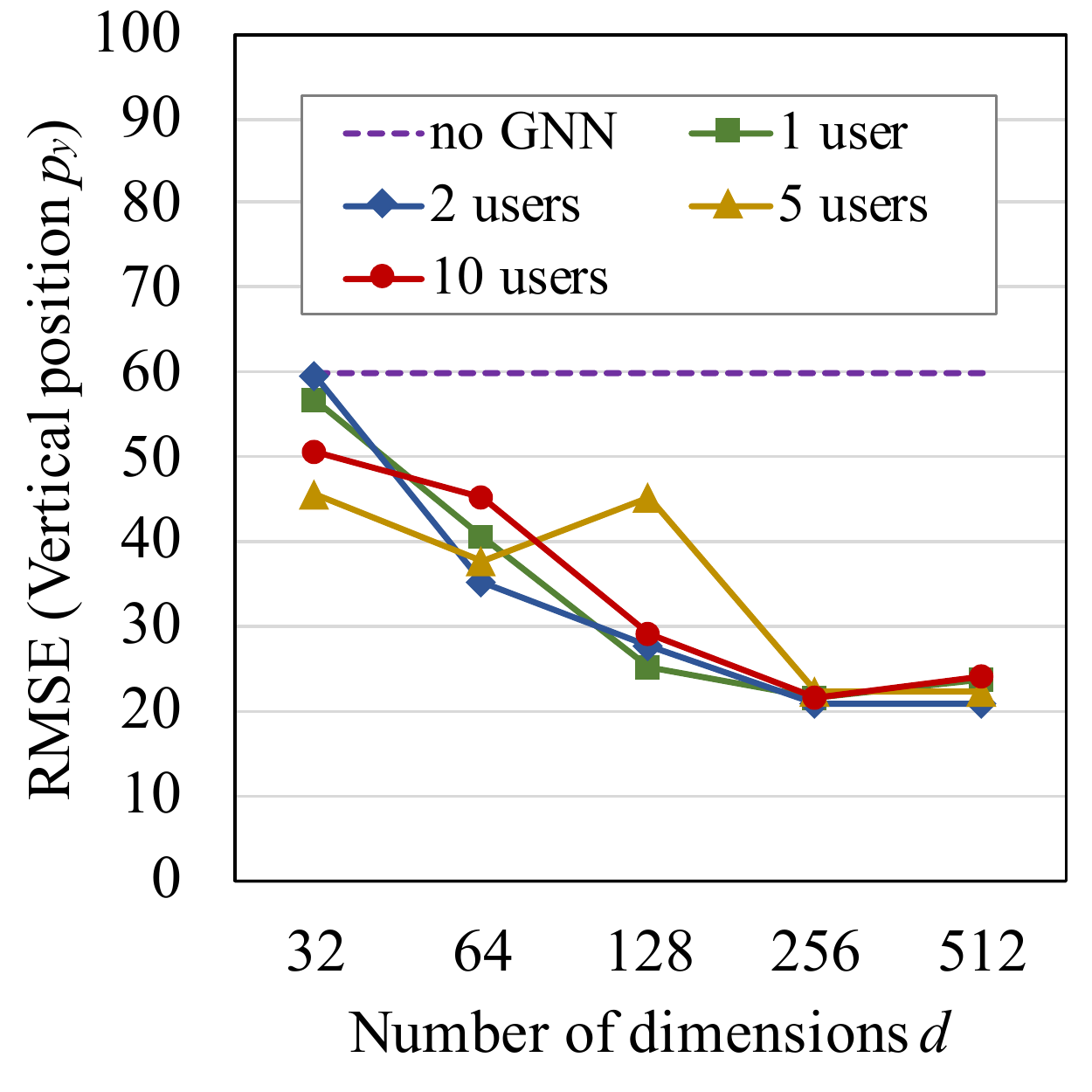} \\
    \includegraphics[height=3.42cm]{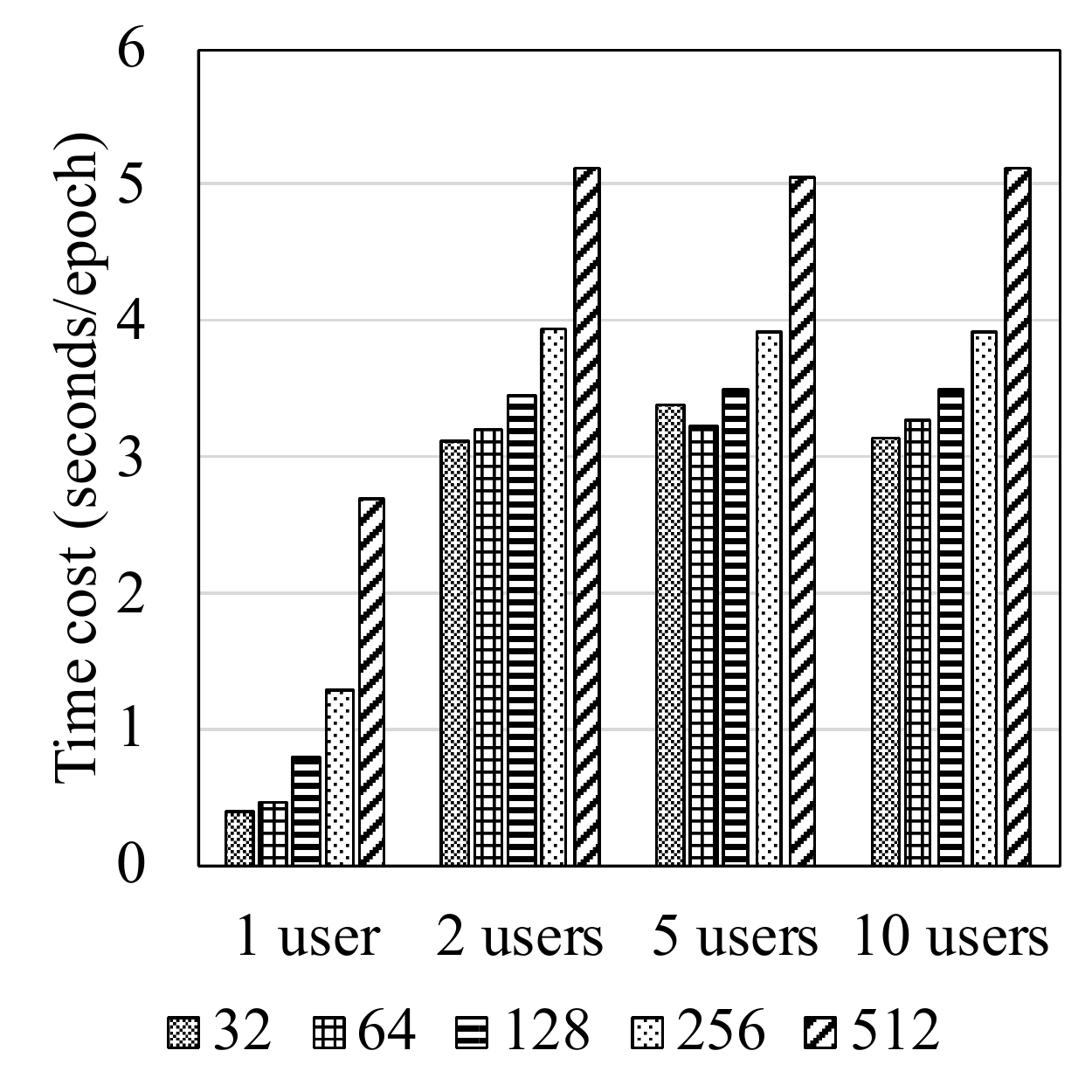}
    \includegraphics[width=3.42cm]{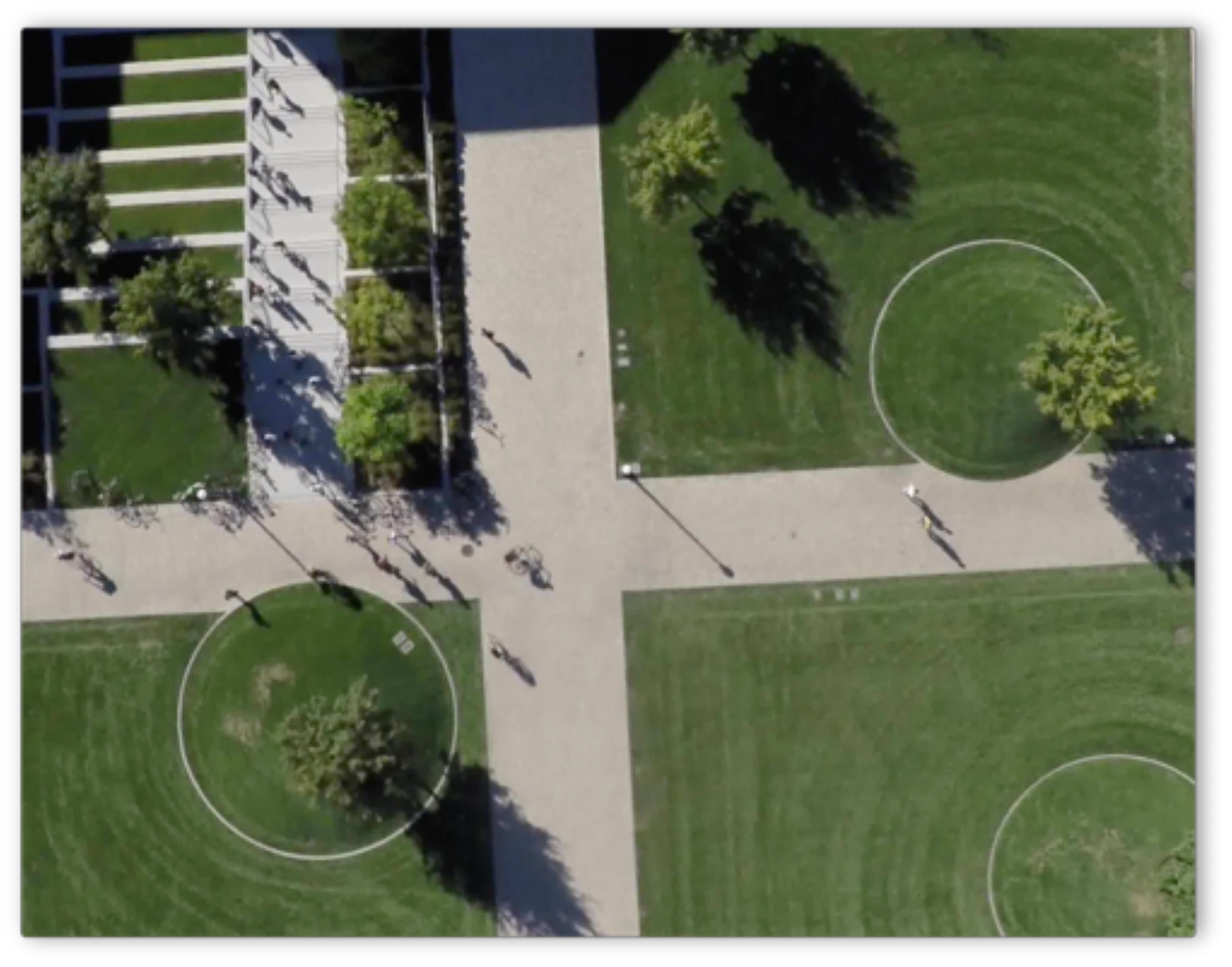}
    \vline
    \includegraphics[height=3.42cm]{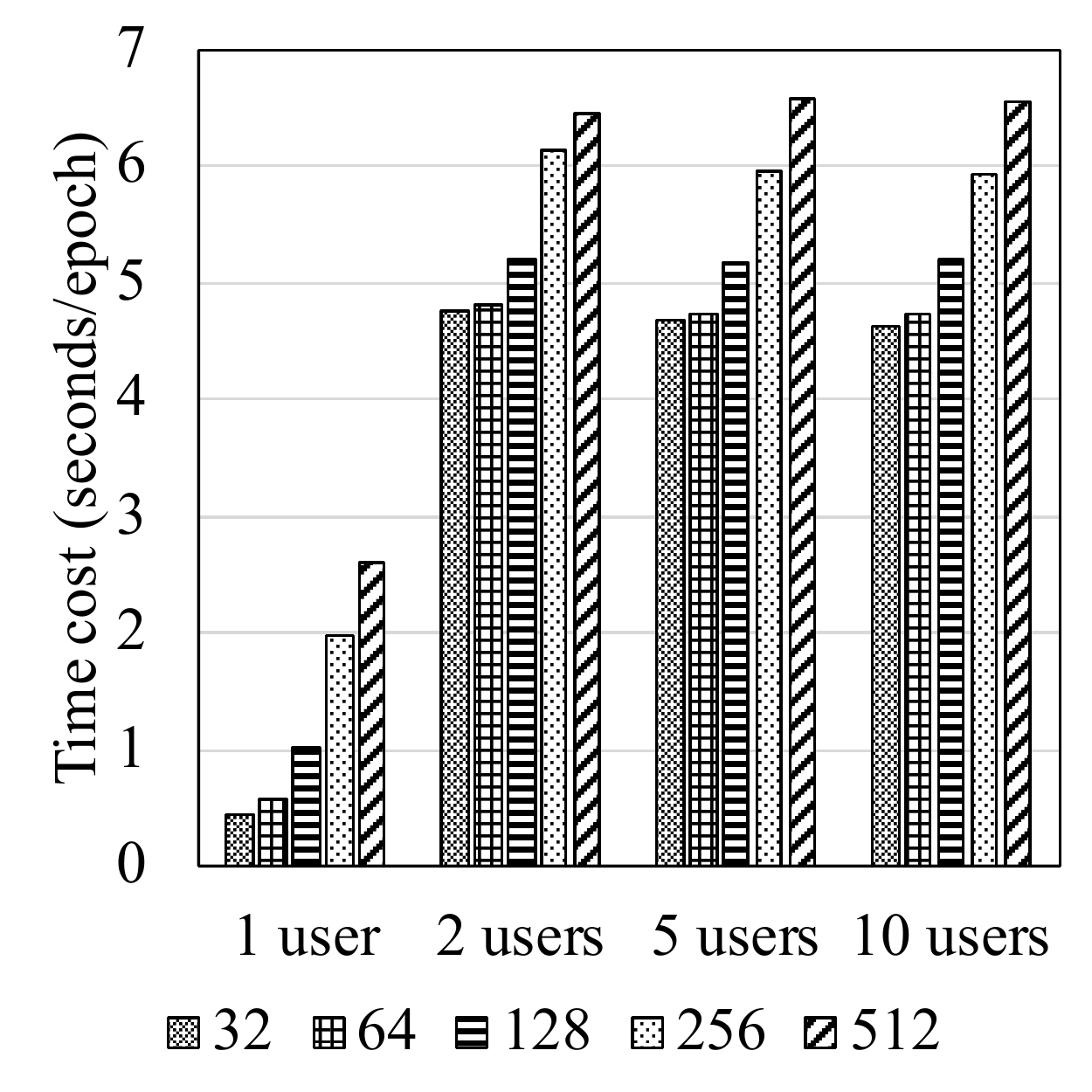}
    \includegraphics[width=3.42cm]{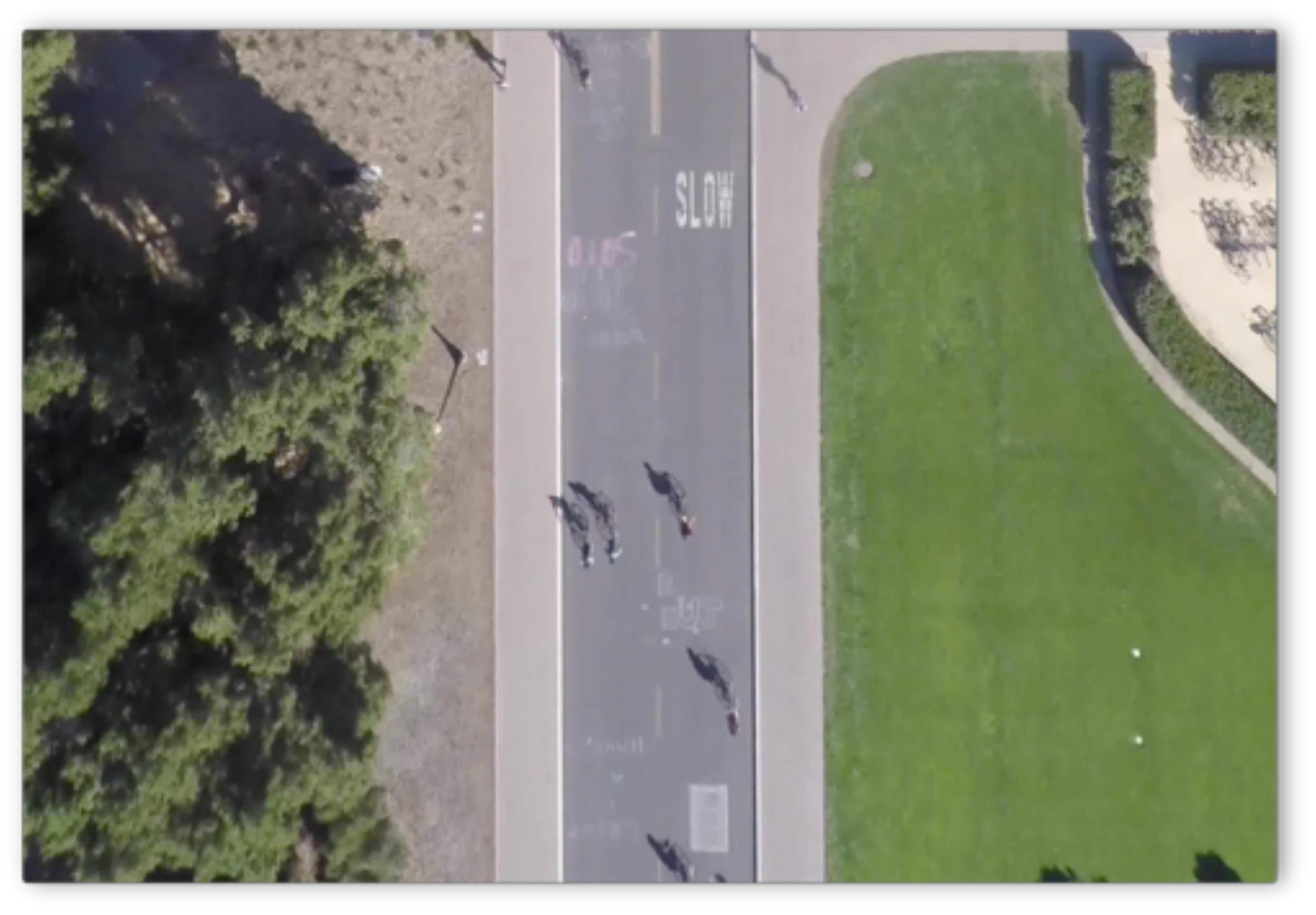} \\
    (c) Huang Y2E2 Buildings (hyang): 1340$\times$1730
    \quad \quad
    (d) Littlefield Center (little): 1322$\times$1945
    \caption{Results on predicting horizontal position $p_x$ (top left), predicting vertical position $p_y$ (top right), and running time (bottom left) as well as an example video frame (bottom right).}
    \label{fig:results}
    \vspace{-0.1in}
\end{figure}

\subsection{Results on efficiency}

From the bar charts in Figure~\ref{fig:results} we observe that when the number of user $m > 1$, which enables federated optimization, the time cost is higher than that of $m=1$. However, it does not become significantly higher as $m$ becomes bigger. Time cost under different $m$ values is comparable. The time cost of non-federated learning with the number of dimensions $d=512$ is close to that of federated learning with $d=32$. Feddy's time cost is $1.5\times$ to $3\times$ of that of non-federated GNN.

\subsection{Results on secure aggregation}

Masking at the user side and the aggregation at the parameter server side are computationally expensive, however the masking and the aggregation for all weights can be computed in parallel. Therefore, we implemented a parallel program for the secure aggregation algorithms, where masking and aggregation are parallelized, and performed a simulation to measure the computation costs. COTS smart cameras are equipped with quad-core processors (e.g., NEON-1040 by ADLINK), therefore we performed the simulation with 4 threads and measured the end-to-end elapsed time for each user. We repeated all simulation for 50 times and measured the average. All parameters are chosen such that we have 112-bit security as recommended by NIST \cite{barker2018transitioning} (the bidwith of $p_\nu$'s is 118 bits, the bitwidth of $N$ is 2048 bits, etc.).

\newcommand{\customwidth}{.245\textwidth}
\begin{figure}[t]\centering
    \includegraphics[width=\customwidth]{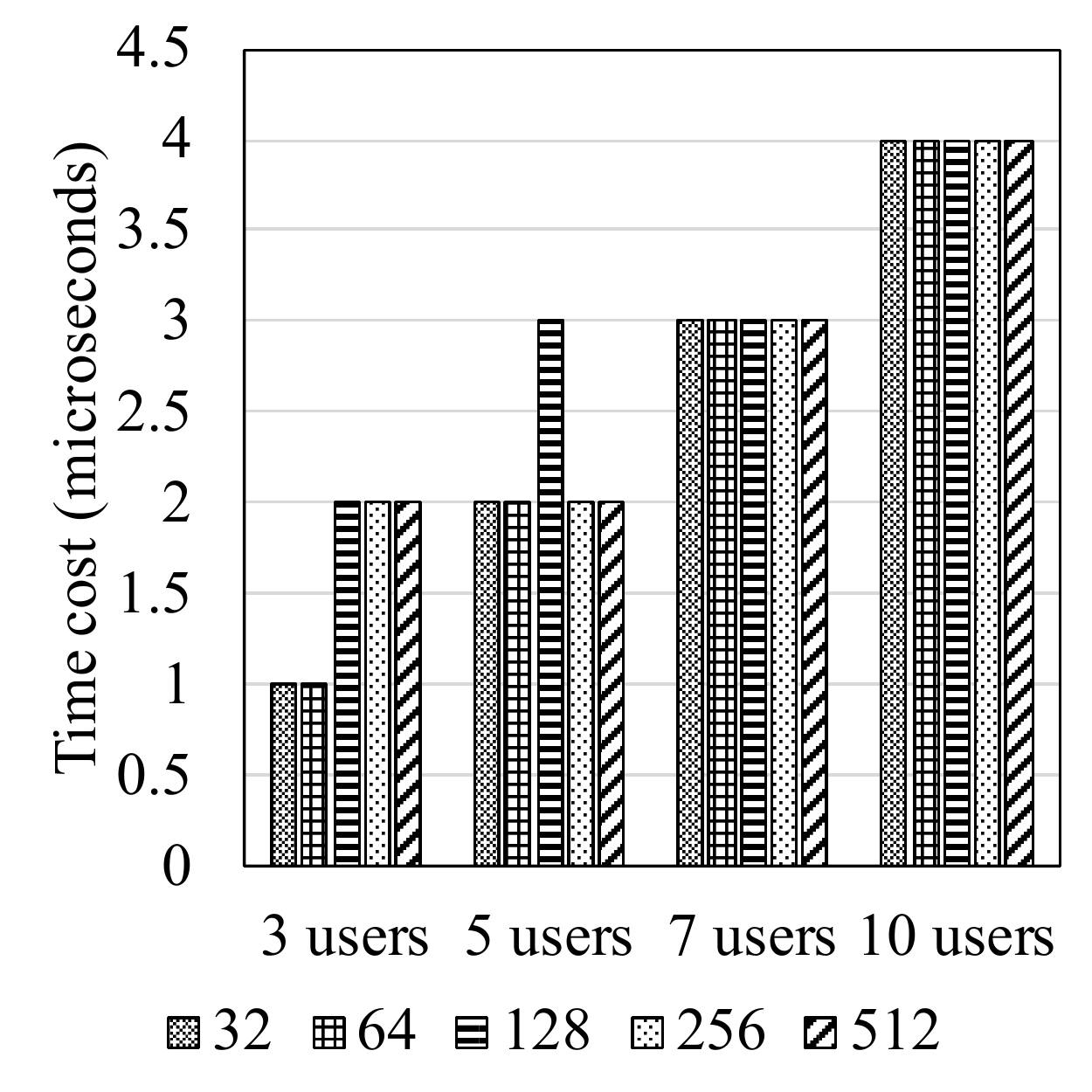}
    \includegraphics[width=\customwidth]{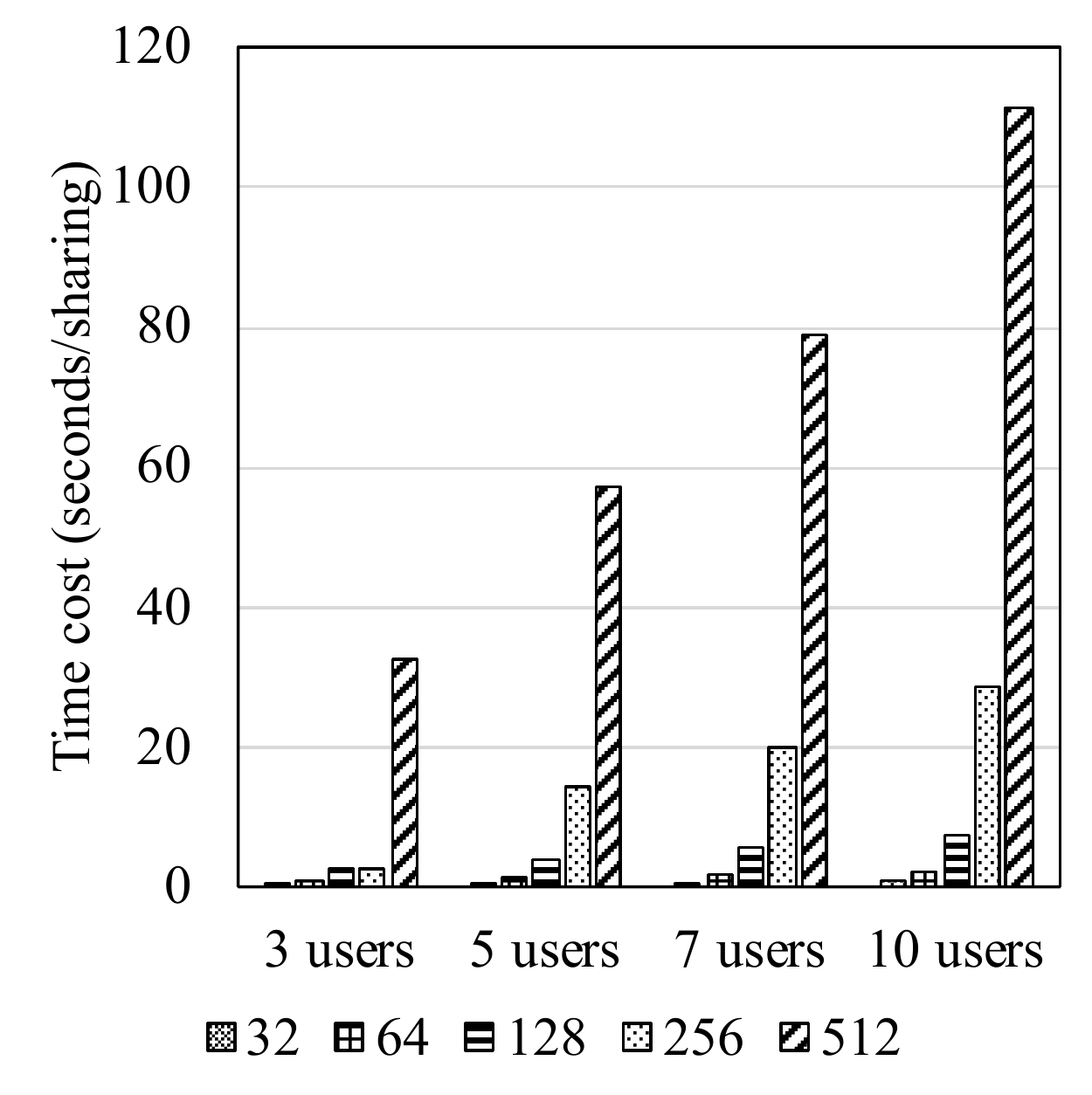}
    \includegraphics[width=\customwidth]{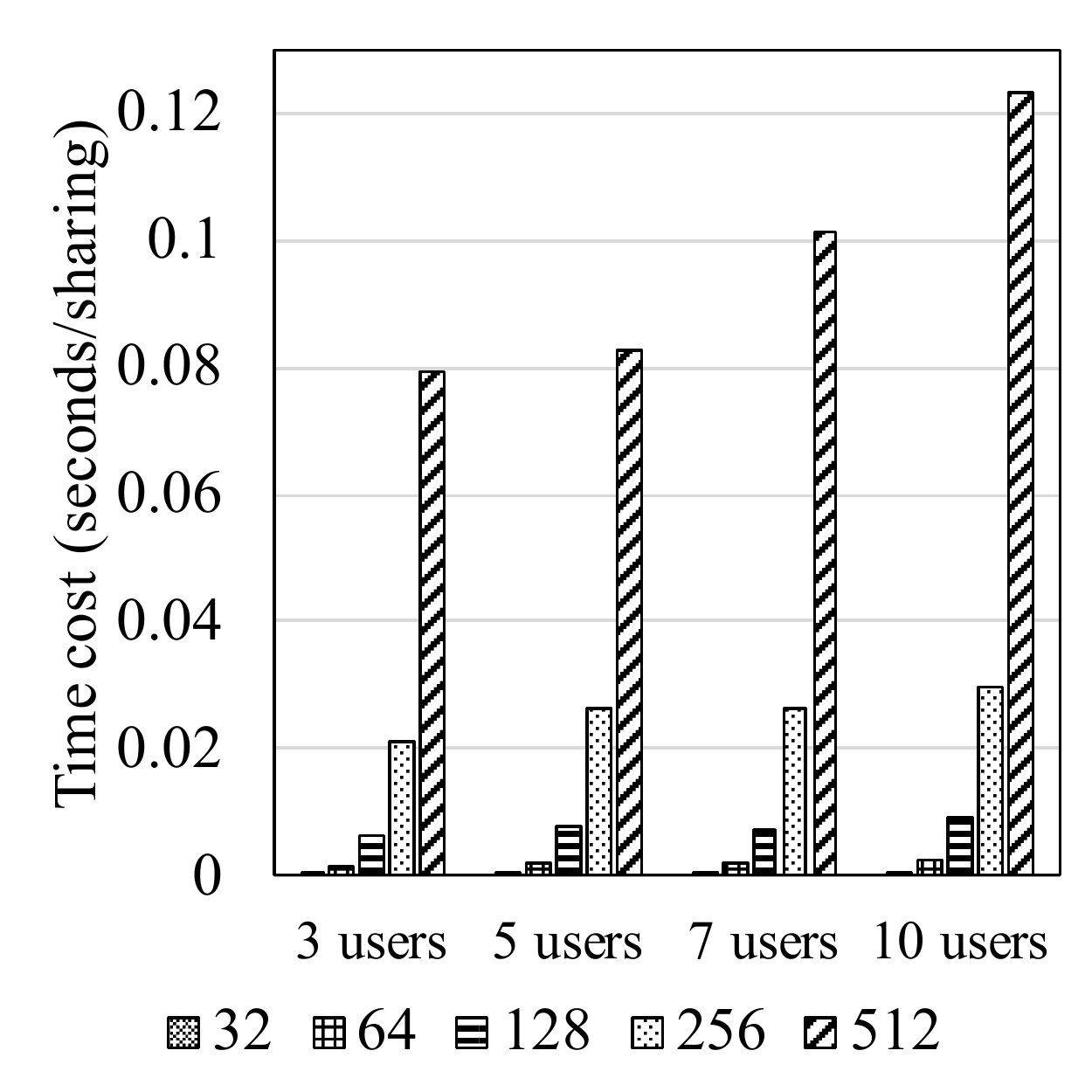}
    \includegraphics[width=\customwidth]{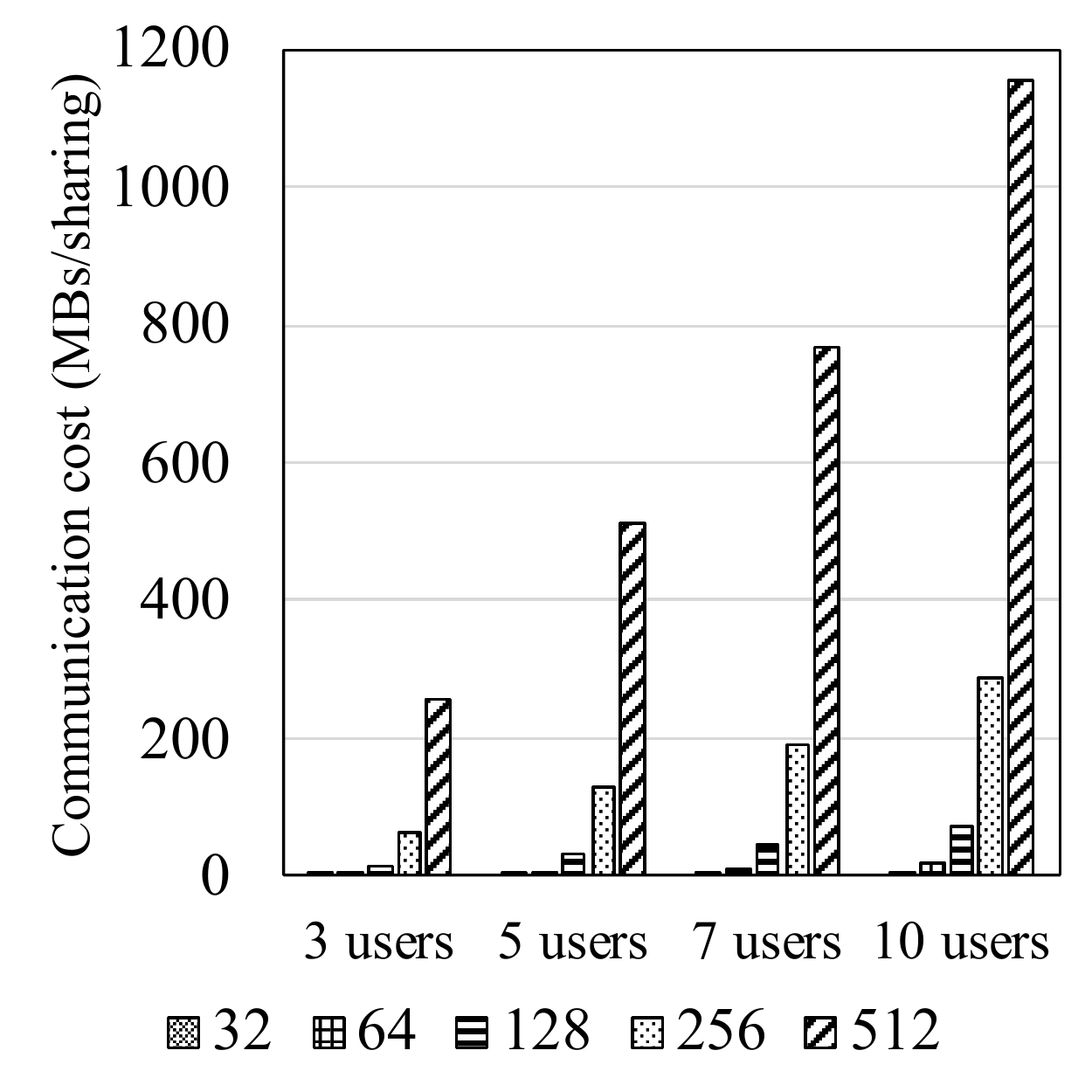}
    \\
         \quad\quad(a) Key Generation \quad\quad\quad
    (b) Masking  \quad\quad\quad\quad (c) Aggregation \quad\quad\quad (d) Communication\\
    \caption{Computation and communication costs of secure aggregation per user}
    \label{fig:results-aggregate}
    \vspace{-0.1in}
\end{figure}

The costs of key generation (Figure~\ref{fig:results-aggregate}(a)) is negligible. Although they grow linearly \textit{w.r.t.} the number of users, the key generation is a one-time process, therefore the overhead of the key generation is negligible. The masking (Figure~\ref{fig:results-aggregate}(b)), however, is not negligible. The costs grow linearly with the number of users and quadratically with the dimension. The size of keys grow linearly \textit{w.r.t.} the number of users (since $p_\nu=\sum_{\mu\in U}p_{\nu,\mu}$) up to the modulus, therefore the costs will not grow after the number of users reaches certain point. However, it is unavoidable to have a quadratic growth because the number of the parameters that need to be masked and shared are quadratic \textit{w.r.t.} the dimension. Still, the costs of masking are acceptable because the aggregation via Feddy occurs once every 10 epochs, and the cost of the masking is comparable to the total costs of the GNN training during those epochs. This is a tradeoff one needs to make to achieve a strong provable data security.
The aggregation performed by the parameter server (Figure \ref{fig:results-aggregate}(c)) is negligible, since the aggregation occurs once every 10 epochs. The costs grow linearly with the number of users, however reasonably strong servers can handle such computation. Finally, we present the total size of the messages one user needs to receive from all other users (Figure~\ref{fig:results-aggregate}(d)). There are non-negligible communication costs caused by the masking, and this is also the tradeoff for the sake of provable security.
\section{Related Work}

\textbf{Dynamic/evolutionary GNN.} GNN models have been developed to learn static graph data \cite{henaff2015deep,niepert2016learning,defferrard2016convolutional,hamilton2017representation,schlichtkrull2018modeling,seo2018structured}. Dynamic graph embeddings are expected to preserve specific structures such as triadic closure processes \cite{zhou2018dynamic}, attribute-value dynamics \cite{li2017attributed}, continuous-time patterns \cite{nguyen2018continuous}, interaction trajectories \cite{kumar2019predicting}, and out-of-sample nodes \cite{ma2018depthlgp}. Graph convolutional networks are equipped with self-attention mechanisms \cite{sankar2018dynamic}, Markov mechanisms \cite{qu2019gmnn}, or recurrent models \cite{pareja2020evolvegcn} for dynamic graph learning \cite{manessi2017dynamic}. We focus on modeling spatial and dynamic information in video graph sequences.

\textbf{Federated machine learning.} As clients have become more powerful and communication efficient, developing deep networks for decentralized data has attracted lots of research \cite{mcmahan2016communication,konevcny2016federated}. Federated optimization algorithms can distribute optimization beyond datacenters \cite{konevcny2015federated}. Bonawitz \emph{et al.} proposed a new system design to implement at scale \cite{bonawitz2019towards,zhu2019multi}. Yang \emph{et al.} presented concepts (forming a new ontology) and methods on employing federated learning in machine learning applications \cite{yang2019federated}.

\textbf{Secure aggregation on time-series data.} Bonawitz \emph{et al.} developed practical secure aggregation for federated learning on user-held data \cite{bonawitz2016practical} and for privacy-preserving machine learning \cite{bonawitz2017practical}. However, those methods cannot be directly applied for time-series or sequential data. Existing work on privacy-preserving aggregation of time-series data needs semantic analysis on dynamic user groups using state-of-the-art machine learning \cite{shi2011privacy,joye2013scalable,jung2016pda}.

\section{Conclusions}

We presented Federated Dynamic Graph Neural Network, a distributed and secure framework to learn the object representations from multi-user graph sequences. It aggregates both spatial and dynamic information and uses a self-supervised loss of predicting the trajectories of objects. It is trained in a federated learning manner. The centrally located server sends the model to user devices. Local models on the respective user devices learn and periodically send their learning to the central server without ever exposing the user's data to the server. We design secure aggregation primitives that protect the security and privacy in federated learning with scalability. Experiments on four real-world video camera datasets demonstrated that Feddy achieves great effectiveness and security.


\section*{Broader impacts}


This paper presents provably secure federated learning based on a novel secure aggregation scheme. Specifically, this paper addresses the data security issues in the federated learning for the GNN, which is a promising deep learning framework for time-series video datasets. Considering the sensitivity of the video data collected by surveillance cameras, the broader impacts of this paper lie in the enhanced data security which leads to enhanced cybersecurity and individual privacy. 

The increased computation and communication costs may negatively impact the broader impacts, however there is active research in applied cryptography for accelerating cryptographic primitives. Therefore, the negative impacts from the increased overhead will be mitigated over time.

\bibliographystyle{named}
\bibliography{ref,gnn,fl,taeho,gcn,evo}

\appendix

\section{Correctness of Feddy}

We rigorously prove that the secure aggregation employed in Feddy correctly let parameter server compute the aggregated values.

\begin{theorem}
Suppose $\sum_{\nu\in U}p_\nu=0\pmod{N}$ and $H$ is defined as $H:T\rightarrow \mathbb{G}_H$, where $\mathbb{G}_H$ is a multiplicative cyclic group of order $N$ with multiplication modulo $N^2$ being the multiplicative group operator. Then, we have:
\begin{displaymath}
    \frac{\big((\prod_{\nu\in U}y_\nu)-1\big)\mod N^2}{N}\mod N=\sum_{\nu\in U}x_\nu.
\end{displaymath}
\end{theorem}

Note that the $H$ above can be easily constructed as $H(t):=(r^{\varphi(N)})^t$ if $r^{\varphi(N)}$ is included in the system-wide parameters. This can be done securely either by employing a crypto server who generates the system-wide parameters and destroys all secret values \cite{jung2016pda}, or by employing a secure multi-party computation protocol that generates RSA keys \cite{damgaard2010efficient}. We adopt the former approach in this paper. 

\begin{proof} We start by simplifying $\frac{((\prod_{\nu\in U}y_\nu)-1)\mod N^2}{N}$.
\begin{displaymath}\begin{split}
    \frac{\big((\prod_{\nu\in U}y_\nu)-1\big)\mod N^2}{N}&= \frac{\Big(\big(\prod_{\nu\in U}(1+N)^{x_\nu}H(t)^{p_\nu}\big)-1\Big)\mod N^2}{N}\\
    &= \frac{\Big((1+N)^{\sum_{\nu\in U} x_\nu}H(t)^{\sum_{\nu\in U}p_\nu}-1\Big)\mod N^2}{N}\\
    &~~~~ \text{(due to the binomial theorem)}
    \\
    &= \frac{\Big((1+N\sum_{\nu\in U} x_\nu)H(t)^{\sum_{\nu\in U}p_\nu}-1\Big)\mod N^2}{N}\\
    &~~~~ \text{(because the order of the group $\mathbb{G}_H$ is $N$)}\\
    &= \frac{\Big((1+N\sum_{\nu\in U} x_\nu)-1\Big)\mod N^2}{N}\\
    &=\frac{\Big((1+N\sum_{\nu\in U} x_\nu)-1\Big)\mod N^2}{N}\\
    &=\frac{N\sum_{\nu\in U}x_\nu \mod N^2}{N}\\
    &=\frac{(N\sum_{\nu\in U}x_\nu) - kN^2}{N}~~\text{for some integer $k$}\\
    &=(\sum_{\nu\in U} x_\nu) -kN ~~\text{for some integer $k$}\\
\end{split}
\end{displaymath}
Then, it follows that $\frac{((\prod_{\nu\in U}y_\nu)-1)\mod N^2}{N}\mod N=\sum_{\nu\in U}x_\nu$. Note that $N$ is much larger than $x_\nu$'s, so $\sum_{\nu\in U}x_\nu$ will not likely exceed the modulus $N$. For example, in our implementation, $N$'s bitwidth is 2048 bits (i.e., as large as $2^{2048}$) while $x_\nu$s' bitwidths are less than 100 bits. In such a parameter setting, we can add at least $2^{1948}$ $x_\nu$'s without exceeding the modulus $N$. Therefore, in practical settings where we are dealing with gradients of the neural networks, we do not need to worry about the result being incorrect due to the modulus overflow issues.
\end{proof}

Then, since $\sum_{\nu\in U}x_\nu$ can be calculated correctly with Feddy, the parameter server can map the integer sum to the real-number sum, which yields the aggregated gradients. Recall that the approximation error caused by integer-real approximation with the fixed-point representation is bounded above by $2^{-(e+1)|U|}$ as described in Section \ref{sec:secure-aggregation}.

\section{Security of Feddy}
Here we present a privacy analysis of the framework inspired by \cite{jung2016pda}.  In general this uses standard techniques for proving indistinguishability (IND-CPA) \cite{joye2013scalable, ananth2015indistinguishability}.

\subsection{Adversary Model}

Note that due to privacy concerns, the result of the analytical computation should  only be given to the aggregator (parameter server), and any user’s data should be kept secret from anyone else but the owner unless it is deducible from the aggregated value. Additionally, both users and the aggregator are assumed to be \textit{semi-honest adaptive} adversaries. Informally, these adversaries will follow the protocol specifications correctly, but they may perform extra computation to try to infer others' private values (i.e., semi-honest), and the computation they perform can be based on their historical observation (i.e., adaptive). For our scenario, if users tamper with the protocol (i.e., not following the protocol specifications correctly), it is highly likely that the aggregator will detect it since the outcome of the protocol will not be in a valid range due to the cryptographic operations on large integers.  However, the aggregator is interested in recovering the correct  result, so they will not be motivated to attempt to maliciously tamper with the protocol. Note that users may report a value with small deviation such that the analytic result still appears reasonable for many reasons (by mistake etc.), but evaluating the reliability of the reported value is beyond the scope of this paper. Also note that adversaries are adaptive in the sense that they may produce their public values adaptively after seeing others’ public values.  We assume all the communication channels are open to anyone (i.e. as a result anyone can overhear/synthesize any message).

\subsection{Security Definition}

To formally define the security of the framework, we present a precise definition of Feddy:

\begin{definition}\label{FeddyDef}
Our Federated Dynamic Graph Neural Network 
($Feddy$) is the collection of the following four polynomial time algorithms: $Setup$, $KeyGen$, $Mask$, and $Aggregate$.

$Setup(1^{\kappa}) \longrightarrow params$ is a probabilistic setup algorithm that is run by a \textit{crypto server} (which is different from the parameter server) to generate system-wide public parameters that define the integer groups/rings the protocols will be operated on, denoted $params$ given a security parameter $\kappa$ as input.

$KeyGen(params) \longrightarrow \{ EK_\nu \}_\nu$ is a probabilistic and distributed algorithm jointly run by the users. Each user $\nu\in U$
will secretly receive his own secret key $EK_\nu$.

$Mask(x_{\nu},EK_\nu,T) \longrightarrow y_\nu=(1+N)^{x_\nu}H(t)^{p_\nu}$ is a deterministic algorithm run by each user $i$ to mask his private value $x_\nu$ into the maksed value $y_\nu$ using $t \in T_f$. The output $y_\nu$ is published in an insecure channel.

$Aggregate({y_\nu| \forall \nu \in U}) \longrightarrow \sum_{\nu\in U}\left(x_{\nu}\right)$ is run by the aggregator to aggregate all the encoded private values ${C(x_{\nu})}$’s to calculate the sum over $\{x_\nu\}_{\nu\in U}$.

\end{definition}

The security of the aggregation is formally defined via a data publishing game (Figure \ref{fig:data-game}), similar to \cite{jung2016pda}.

\begin{figure}
    \centering
    \begin{mdframed}

\textbf{Setup:)} $3$ disjoint time domains are chosen: $T_1$ for phase
$1$, $T_2$ for phase $2$, and $T_c$ for the challenge phase.

\textbf{Init:)} The adversary declares their role in the scheme (i.e., aggregator or user), and the challenger controls the remaining users. The users engage in the key generation.

\textbf{Phase 1 in $T_1$:)} The adversary submits polynomially many queries to the \textit{masking oracle}\footnote{The masking oracle's role is to return the masked values when $x$, $t$, and $\nu$ are given.} and receive the masked values for any $x$ and any time window $T\subseteq T_1$ and any user in $U$ including those are not adversaries. If the declared time windows do not overlap with each other, the masking oracles returns all masked values to the adversary; otherwise, adversaries receive nothing.

\textbf{Challenge in $T_c$:)} The adversary declares the target time window $T_c$. Then, they submit two sets of values $\{x_{\nu,1}\},\{x_{\nu,2}\}$, such that $\sum_{\nu\in U}x_{\nu,1} = \sum_{\nu\in U}x_{\nu,2}$, to the challenger. The challenger flips a fair binary coin $b$ and generate the masked values $\{y_{\nu,b}\}$ based on $x_{\nu,b}$, which are given to the adversary.

\textbf{Phase 2 in $T_2$:)} Phase 1 is repeated adaptively, but the time window $T$ should be a subset of $T_2$.

\textbf{Guess:)} The adversary gives a guess $b'$ on $b$. The advantage of the adversary in this game is defined $\textit{adv}_{A} = | \mathbf{Pr}[b' - b] - \frac{1}{2}|$. 
\end{mdframed}
    \caption{Data Publishing Game}
    \label{fig:data-game}
\end{figure}

We define the security of the masking in $Feddy$ as follows:

\begin{definition}
\label{MaskDef}
The random masking in the $Feddy$ is indistinguishable against the chosen-plaintext attack (IND-CPA) if all polynomial time adversaries’ advantages in the game are of a negligible function \textit{w.r.t.} the security parameter $\kappa$ when $T_1,T_2$, and $T_c$ are three disjoint time domains.

\end{definition}
 
  Next, we present a standard simulation-based definition of the security \cite{lindell2017simulate, boyle2017group, gordon2015constant}, that is achieved by our Feddy protocol. Note that Feddy is not an encryption scheme, and although we leverage Definition \ref{MaskDef} to prove security later, techniques such as proving IND-CCA or IND-CPA alone do not directly demonstrate the security of the entire protocol.  Informally speaking, the $Feddy$ scheme is private if adversaries do not gain more information than the input that they control, the output, and what can be inferred from each of them. 

We define the security of $Feddy$ as follows:

\begin{definition}\label{secure-def}
The aggregation scheme $Feddy$ for a class of summation functions $F$ is said to be private for $F$ against semi-honest adversaries if for any $f \in F$ and for any probabilistic polynomial time adversary $\mathcal{A}$ controlling a subset $A$ of all players, there exists a probabilistic polynomial-time simulator $\mathcal{S}$ such that for any set of inputs $X := (x_1, \cdots,  x_n)$ in the domain of $f$ where the $i$-th player $P_i$ controls $x_i$,
\begin{displaymath}
  \{\mathcal{S}(f(X), A, \{x_j \mid P_j \in A\})\}_{\kappa} \indist 
  \{\mathsf{View}_A ^{Feddy} (X)\}_\kappa
\end{displaymath}
where $\indist$ refers to computational indistinguishability, $\kappa$ is the security parameter, and $\mathsf{View}_A ^{Feddy} (X)$ represents the messages received by members of $A$ during execution of protocol $Feddy$.
\end{definition}

\subsection{Security Proof}

Before we present the security proofs, we present the computational problems and the hardness assumptions Feddy relies on.

\begin{definition}
Decisional Diffie-Hellman (DDH) problem in a group $\mathbb{G}$ with generator $g$ is to decide whether $g^c=g^{ab}$ given a triple $(g^a,g^b,g^c)$, where $a,b,c\in\mathbb{Z}$. An algorithm $\mathcal{A}$'s advantage in solving the DDH problem is defined as
\begin{displaymath}\begin{split}
adv_{\mathcal{A},\mathbb{G}}^{DDH}&=\Big|\mathbf{Pr}\big[1\leftarrow\mathcal{A}(g^a,g^b,g^{ab}\in\mathbb{G})\big]\\
&-\mathbf{Pr}\big[1\leftarrow\mathcal{A}(g^a,g^b,g^c\in\mathbb{G})\big|c\leftarrow_R \mathbb{Z} \big]\Big|
\end{split}\end{displaymath}
where $1\leftarrow \mathcal{A}(\cdot)$ if the algorithm outputs `yes' and 0 otherwise, and the probabilities are taken over the uniform random selection $c\leftarrow_R \mathbb{Z}$ as well as the random bits of $\mathcal{A}$.
\end{definition}

\begin{definition}
Decisional Composite Residuosity (DCR) problem in $\mathbb{Z}_{N^2}^*$ is to decide whether a given element $x\in\mathbb{Z}_{N^2}^*$ is an $N$-th residue modulo $N^2$ or not.
\end{definition}

The DDH problem in $\mathbb{Z}_N^*$ and the DCR problem are widely belived to be intractable \cite{jung2016pda,joye2013scalable,boneh1998decision}. We will prove the security of Feddy by proving the following theorem.

\begin{theorem}

With the assumptions that the DDH problem is hard in $\mathbb{Z}_N^*$ and that the DCR problem is hard, the random masking in our Feddy scheme is indistinguishable against chosen-plaintext attacks (IND-CPA) under the random oracle model. Namely, for any PPTA $A$, its advantage $adv_A$ in the data publishing game
is bounded as follows:

\[ a d v_{\mathcal{A}} \leq \frac{e\left(q_{c}+1\right)^{2}}{q_{c}} \cdot a d v_{\mathcal{A}}^{D D H} \]

where $e$ is the base of the natural logarithm, $q_{c}$ is the number
of adversaries' queries submitted to the masking oracle, and $adv_\mathcal{A}^{DDH}$ is the advantage in solving the DDH problem. Note that $adv_\mathcal{A}^{DDH}$ is negligibly small since DDH problem is widely believed to be hard.

\end{theorem}

We prove the theorem by adapting the proofs from \cite{jung2016pda,joye2013scalable}.

\begin{proof}\label{MaskProof}

To prove the theorem, we present three games Game $1,$ Game $2,$ and Game $3,$ in which we use $\mathcal{A}$ and $\mathcal{B}$ to denote the adversary and the challenger. For each $l \in\{1,2,3\},$ we denote $E_{l}$ as the event that $\mathcal{B}$ outputs 1 in the Game $l$, and we define $a d v_{l}=\left|\mathbf{Pr}\left[E_{l}\right]-\frac{1}{2}\right|$

\textbf{Game 1:)} This game is exactly identical to the earlier data publishing game. $\mathcal{A}$ 's masking queries $(T,\left\{x_{\nu}\right\}_{\nu})$  are answered by returning the masked values $\{y_{\nu}\}_{\nu} .$ In the challenge phase, the adversary $\mathcal{A}$ the associated time window $T_{c}$, and two sets of values $\{x_{\nu,1}\},\{x_{\nu,2}\}$ which satisfy $\sum_{\nu\in U}x_{\nu,1} = \sum_{\nu\in U}x_{\nu,2}$. Then, the challenger $\mathcal{B}$ returns the corresponding masked values to the adversary $\mathcal{A}$. When the game terminates, $\mathcal{B}$ outputs 1 if $b^{\prime}=b$ and 0 otherwise. By definition, $adv_{1}=\left|\mathbf{P r}\left[E_{1}\right]-\frac{1}{2}\right|=a d v_{\mathcal{A}}$

\textbf{Game 2:)} In Game 2 , the adversary $\mathcal{A}$ and the challenger $\mathcal{B}$ repeat the same operations as in Game 1 using the same time windows of those operations. However, for each masking query in Game 1 at time $t \in T_{1} \cup T_{2}$, the challenger $\mathcal{B}$ flips a biased binary $\operatorname{coin} \mu_{T}$ for the entire time window $T$ which takes 1 with probability $\frac{1}{q_{c}+1}$ and 0 with probability $\frac{q_{c}}{q_{c}+1}$. When the Game 2 terminates, $\mathcal{B}$ checks whether any $\mu_{T}=1 .$ If there is any, $\mathcal{B}$ outputs a random bit. Otherwise, $\mathcal{B}$ outputs 1 if $b^{\prime}=b$ and 0 if $b^{\prime} \neq b .$ If we denote $F$ as the event that $\mu_{T_{f}}=1$ for any $T_{f},$ the  analysis in \cite{coron2000exact} shows that $\operatorname{Pr}[\bar{F}]=\frac{1}{e\left(q_{c}+1\right)} .$ According to \cite{joye2013scalable} Game 1 to Game 2 is a transition based on a failure event of large probability, and therefore we have $adv_{2}=adv_{1} \operatorname{Pr}[\bar{F}]=\frac{a d v_{1}}{e\left(q_{c}+1\right)}$.

\textbf{Game 3:)} In this game, the adversary $\mathcal{A}$ and the challenger $\mathcal{B}$ repeat the same operations as in Game 1 using the same time windows of those operations. However, there is a change in the answers to the masking queries $\left(T,\left\{x_{\nu}\right\}_{\nu}\right).$ The oracle will respond to the query with the following masked values:
\[
\forall \nu: y_\nu=\left\{\begin{array}{ll}
(1+N)^{x_{\nu}} H\left(t \right)^{p_\nu} & \mu_{T}=1 \\
(1+N)^{x_{\nu}}\left(H\left(t\right)^{s}\right)^{p_\nu} & \mu_{T}=0
\end{array}\right.
\]
where $s$ is a uniform randomly chosen element from $\mathbb{Z}_{N}$ that is fixed for the same aggregation.  When Game 3 terminates, $\mathcal{B}$ outputs 1 if $b^{\prime}=b$ and 0 otherwise. Due to Lemma 8 from \cite{jung2016pda}, distinguishing Game 3 from Game 2 is at least as hard as a DDH problem in $\mathbb{Z}_N$ for any adversary $\mathcal{A}$ in Game 2 and Game 3. It follows then:
$\left|\operatorname{Pr}\left[E_{2}\right]-\operatorname{Pr}\left[E_{3}\right]\right| \leq a d v_{\mathcal{A}}^{D D H} .$ The answers to masking queries in Game 3 are identical to those of Game 2 with probability $\frac{1}{q_{c}+1}$ and different from those of Game 2 with probability $\frac{q_c}{q_{c}+1}$. In the latter case, due to the random element $s$, $y_\nu$ is uniformly distributed in the subgroup $\langle H(t)\rangle$ which completely blinds $(1+N)^{x_\nu}$ and $\mathcal{A}$ can only randomly guess $b^{\prime}$ unless he can solve the DCR problem with non-negligible advantages (which is false under our assumption that the DCR problem is hard). Then, $b^{\prime}=b$ with probability $1 / 2,$ and the total probability $\operatorname{Pr}\left[E_{3}\right]=\frac{\operatorname{Pr}\left[E_{2}\right]}{q_{c}+1}+\frac{q_{c}}{2\left(q_{c}+1\right)} .$ Then, we have:
\[
\begin{aligned}
\left|\operatorname{Pr}\left[E_{2}\right]-\operatorname{Pr}\left[E_{3}\right]\right| &=\left|\left(\operatorname{Pr}\left[E_{2}\right]-1 / 2\right) \cdot \frac{q_{c}}{q_{c}+1}\right|
&=a d v_{2} \cdot \frac{q_{c}}{q_{c}+1} \leq a d v_{\mathcal{A}}^{D D H}
\end{aligned}
\]
Combining the above inequality with the advantages deduced from Game 1 and Game 2 , we finally have:
\[
a d v_{2} \cdot \frac{q_{c}}{q_{c}+1}=a d v_{\mathcal{A}} \cdot \frac{q_{c}}{e\left(q_{c}+1\right)^{2}} \leq a d v_{\mathcal{A}}^{D D H}
\]
thus completing the proof.

\end{proof} 

With the above, we are now ready to prove the security of $Feddy$ in the well known simulation security model \cite{lindell2017simulate} in the following theorem:

\begin{theorem}
Assuming $|\mathcal{C}|$ $= O(\log \kappa)$, all parties properly instantiate and use the $Feddy$ algorithms in Definition \ref{FeddyDef}, and the random masking in Definition \ref{MaskDef} is used, our scheme $Feddy$ privately computes according to Definition \ref{secure-def}
\end{theorem}

\begin{proof}
Without loss of generality, assume the adversary controls the first $m < n$ variables, where $n$ is the total number of all variables in the scheme (i.e. each person $j$ participating in the protocol inputs a collection of private variables (e.g. $x_{j,1}, \cdots, x_{j,r_j}$) where user $j$ controls $r_j$ variables and the total number of variables for all participants sums to $n$ variables).  We show that a probabilistic polynomial time simulator can generate an entire simulated view, given $z=f(x_1, \cdots, x_m, x_{m+1}, \cdots, x_n)$ and $x_1, \cdots x_m$ for an adversary indistinguishable from the view an adversary sees in a real execution of the protocol.  Note the simulator is able to find $x_{m+1}', \cdots, x_n'$ such that $z = f(x_1', \cdots, x_m',$ $x_{m+1}', \cdots, x_n')$ in polynomial time since $|\mathcal{C}| = O(\log \kappa)$.  Besides this, the adversary follows the protocol as described in Figure \ref{FeddyDef}, pretending to be honest.  Note that in any world (i.e. $Ideal$ or $Real$), the adversary can only compute the function over the fixed values submitted by the honest users, because the adversary can only access the fully aggregated set of submitted values.  Individual values submitted by the honest players are secure by the IND-CPA property of the underlying masking scheme as shown in the proof of Theorem \ref{MaskProof}. Now the simulator $\mathcal{S}$ generates a view indistinguishable from that of a real execution, since all parameters broadcast i.e. $x'_\nu$ masked by $p'_\nu$, are indistinguishable from the corresponding ones in the real protocol i.e. $x_\nu$ masked by $p_\nu$ as they are generated identically and have the exact same distribution i.e. a uniformly random distribution.  More specifically all masked values are indistinguishable in both worlds as both sets of masks are generated at random and will have the same random distribution, i.e. the masked values in both worlds will be indistinguishable from random, and thus indistinguishable from each other. Recall that the aggregator does not send the result to users (note that if the users tamper with their submissions the aggregator will likely detect it since the outcome of the protocol will not be in a valid range; the aggregator can abort if necessary depending on the use case of the protocol).  Similarly, masked values sent to the aggregator in both worlds will also be indistinguishable by the assumption that the users correctly follow the protocol, which utilizes IND-CPA secure masking, and the fact that we chose inputs that give the same output. By the security of the DCR and DDH problems, no information can be gained by an adversary intercepting messages. This demonstrates that for the class of functions $F$, our protocol is secure against adversaries $A$ since:
\[
  \{\mathcal{S}(f(X), A, \{x_j \mid P_j \in A\})\}_\kappa \indist 
  \{\mathsf{View}_A ^{Feddy} (X)\}_\kappa
\]

Therefore, the adversary cannot distinguish between real and simulated executions and our protocol securely computes as defined in Definition \ref{secure-def}. \end{proof}

\end{document}